\documentclass[a4paper,english]{lipics-report}
 \author[1]{Paul Gastin}
 \author[2]{Shankara Narayanan Krishna}
 \affil[1]{LSV, ENS Paris-Saclay \& CNRS, Universit\'e Paris-Saclay, France\\
   \texttt{paul.gastin@lsv.fr}}
 \affil[2]{Department of Computer Science \& Engineering, IIT Bombay, India\\
   \texttt{krishnas@cse.iitb.ac.in}}

\authorrunning{P. Gastin and S. Krishna}%mandatory. First: Use abbreviated first/middle names. Second (only in severe cases): Use first author plus 'et al.'
\title{Unambiguous Forest Factorization\footnote{Partly supported by ReLaX, 
UMI2000 (CNRS, ENS Paris-Saclay, Univ.\ Bordeaux, CMI, IMSc).}}
\titlerunning{Unambiguous Forest Factorization}%optional, please use if title is longer than one line
\Copyright{P. Gastin and S. Krishna}%mandatory, please use full first names. LIPIcs license is "CC-BY";  http://creativecommons.org/licenses/by/3.0/
% mandatory: Please choose ACM 2012 classifications from https://www.acm.org/publications/class-2012 or https://dl.acm.org/ccs/ccs_flat.cfm . E.g., cite as "General and reference $\rightarrow$ General literature" or \ccsdesc[100]{General and reference~General literature}. 
%mandatory

\usepackage{babel}
\usepackage{tikz}
\usetikzlibrary{fit} % fitting shapes to coordinates
\usetikzlibrary{backgrounds} % drawing the background after the foreground
\usetikzlibrary{patterns}
\usetikzlibrary{positioning}
\usetikzlibrary{arrows, decorations.markings}

\tikzstyle{vecArrow} = [decoration={markings,mark=at position
   1 with {\arrow[semithick,out=50,in=20]{triangle 60}}},
   shorten >= 5.5pt,
   ]

\tikzstyle{background}=[rectangle,fill=gray!10, inner sep=0.1cm, rounded corners=0mm]
\usepgflibrary{shapes}
\usetikzlibrary{snakes,automata}
\usetikzlibrary{shadows}
\tikzstyle{cir}=[draw=violet, fill=violet!20!white,circle,minimum size=1.5em,inner sep=0em]                                                                                                      
     \tikzstyle{cir1}=[draw=green!50!blue, fill=green!40!white, circle,minimum size=1.5em,inner sep=0.1em]                                                                                                      
\tikzstyle{cir2}=[draw=red,fill=red!20!white,circle,minimum size=1.5em,inner sep=0em]                                                                                                    
\tikzstyle{cir11}=[draw=green!20!blue,fill=green!40!white,rounded rectangle,minimum size=1.5em,inner sep=0em]                                                                                                    
\tikzstyle{cir22}=[draw=green!70!violet,fill=green!15!violet!40!white,rounded rectangle,minimum size=1.5em,inner sep=0em]

\tikzstyle{background}=[rectangle,fill=gray!10, inner sep=0.1cm, rounded corners=0mm]
\tikzstyle{loc}=[draw,rectangle,minimum size=1.4em,inner sep=0em]
\tikzstyle{trans}=[-latex, rounded corners]
\tikzstyle{trans2}=[-latex, dashed, rounded corners]
\newcommand{\Aa}{\mathcal{A}}
\newcommand{\Bb}{\mathcal{B}}

\usepackage{microtype}%if unwanted, comment out or use option "draft"

\usepackage{enumitem}
\usepackage{graphicx}
\usepackage{xspace}

\usepackage{todonotes}

\newcounter{todocounter}

\newcommand\Lang[1]{\mathcal{L}(#1)}

\newcommand{\unitS}{\mathbf{1}_S}

\newcommand{\A}{\ensuremath{\mathcal{A}}\xspace}

\newcommand{\da}{{\downarrow}}
\renewcommand\phi{\varphi}
\makeatletter
\newcommand{\xRightarrow}[2][]{\ext@arrow 0359\Rightarrowfill@{#1}{#2}}
\makeatother

\begin{document}

\maketitle
\abstract
In this paper, we look at an unambiguous version of Simon's forest factorization
theorem, a very deep result which has wide connections in algebra, logic and
automata.  Given a morphism $\varphi$ from $\Sigma^+$ to a finite semigroup $S$,
we construct a universal, unambiguous automaton $\Aa$ which is ``good'' for
$\varphi$.  The goodness of $\Aa$ gives a very easy proof for the forest
factorization theorem, providing a Ramsey split for any word in
$\Sigma^{\infty}$ such that the height of the Ramsey split is bounded by
the number of states of $\Aa$.  An important application of synthesizing good
automata from the morphim $\varphi$ is in the construction of regular transducer
expressions (RTE) corresponding to deterministic two way transducers.

\section{Introduction}

In this paper, we revisit Simon's forest factorization theorem, a central result
in algebraic automata theory.  In his seminal paper \cite{Simon_1990}, Simon
showed that, given a semigroup morphism $\varphi\colon \Sigma^+ \to S$, any word
$w \in \Sigma^{+}$ admits a factorization tree $T(w)$ of height $\leq 9|S|$.
Leaves of $T(w)$ are letters from $\Sigma$ and the yield of $T(w)$ is the word
$w$.  Internal nodes have arity at least two.  Each node $x$ of $T(w)$ is labeled
$F(x)=\varphi(u_x)$ where $u_x$ is the yield of the subtree rooted at $x$.  The main
constraint is that, if an internal node $x$ has arity $n>2$ with children
$x_1,\ldots,x_n$ then $F(x_1)=\cdots=F(x_n)=e$ is an idempotent of $S$. There 
are no constraints for binary nodes. Simon's factorization theorem has many 
deep applications, see e.g., 
\cite{DBLP:conf/dlt/Bojanczyk09,ColcombetFactForest}.

An easy consequence of Simon's forest factorization theorem is that there is a
regular expression $F=\bigcup_{i}F_i$ which is universal (the denoted language
is $\Lang{F}=\Sigma^{+}$) and such that (1) for each subexpression $E$ of some
$F_i$ the denoted language $\Lang{E}$ is mapped by $\phi$ to a single semigroup
element $s_E$, and (2) for each subexpression $E^{+}$ of some $F_i$ the
associated element $s_E$ is an idempotent of $S$.  In addition, the
subexpressions $F_i$ do not use union and have $(\cdot,+)$-depth at most $9|S|$
(the depth of $F_i$ is the longest chain of concatenations and Kleene plus,
i.e., the height of the syntax tree of $F_i$).  A similar statement is given in
\cite{DBLP:conf/fct/Colcombet07,DBLP:conf/dlt/Bojanczyk09,ColcombetFactForest}.
Actually, the converse is also true.  If $F=\bigcup_{i}F_i$ is a universal
regular expressions satisfying (1) and (2), each word $w\in\Sigma^{+}$ can be
parsed according to some $F_i$ and the parse tree is a factorization tree for
$w$.

In this paper, we show how to construct a universal regular expression
$F=\bigcup_{i}F_i$ satisfying (1) and (2) and which in addition is
\emph{unambiguous}.  Therefore, each word $w\in\Sigma^{+}$ admits a
\emph{unique} parse tree according to $F$, which is indeed a factorization tree.
The forest factorization theorem was extended to infinite words by Colcombet in
\cite{DBLP:conf/fct/Colcombet07,Colcombet_2010}.  We also extend our unambiguous
version to infinite words: we can construct an $\omega$-regular expression
$\bigcup_{i}F_iG_i^{\omega}$ which is universal, unambiguous, and the
subexpressions $F_i$, $G_i$ satisfy (1) and (2).  We call these \emph{good}
expressions.

This work is motivated by \cite{lics18} in which regular transducer expressions
(RTE) are  defined and shown equivalent to deterministic two-way transducers
(both for finite and infinite words in which case the transducer may use regular
look-aheads).  The universal good expression is used to parse the input word,
and from the parse tree, the output is suitably computed.  Since deterministic
transducers define functions, it is essential that each input word has a unique
parse tree.  This explains the need for an unambiguous extension of Simon's
forest factorization theorem.  The other properties (1) and (2) are also
essential to compute an RTE equivalent to the given deterministic transducer.
We believe that the existence of \emph{good} regular expressions may have
several other applications.

After the initial bound of $9|S|$ by Simon, there have been follow ups.  In
\cite{DBLP:conf/fct/Colcombet07}, Colcombet extended Simon's result to infinite
words and reduced the bound to $3|S|$.  He used a new proof technique,
constructing Ramsey splits from which the factorization trees can be easily
derived.  Kufleitner \cite{mfcs08} also improved the bound on the height to
$3|S|-1$.  A variant of Kufleitner's proof can be found in
\cite{DBLP:conf/dlt/Bojanczyk09}.  The bound on the height of factorization
trees was further improved in \cite{ColcombetFactForest} to $3|N(S)|-1$, where
$N(S)$ is the maximum over all chains $D_1 <_{\mathcal{J}} \dots <_{\mathcal{J}}
D_k$ of $\mathcal{D}$-classes of the sum $\sum_{i=1}^k N(D_i)$ and $N(D)$ is 1
if $D$ is irregular, else $N(D)$ is the number of elements of $D$ which are
$\mathcal{H}$-equivalent to an idempotent.  The proofs above are based on
Green's relations.  Subsequently, a simplified proof not based on Green's
relations was given in \cite{Diekert_2016} using the local divisor technique.
Also in \cite{ColcombetFactForest}, a deterministic version of Simon's forest 
factorization is given, but to achieve the determinism, conditions (1) and (2) 
had to be weakened.

The main contributions of this paper are as follows.  Given a semigroup morphism
$\varphi\colon \Sigma^+ \to S$, we construct a universal, unambiguous automaton
$\Aa$ that we call ``$\varphi$-good''.  The goodness of $\Aa$ is determined by
the following conditions (i) it is unambiguous and universal (it accepts all
words in $\Sigma^{\infty}$), (ii) it has a unique initial state $\iota$ with no
incoming transitions to it, (iii) it has a unique final state $f$ with no outgoing
transitions from it, (iv) there is a total ordering on the states of the
automaton such that $Q\backslash \{\iota, f\}< f < \iota$, and (v) for each
state $q$, the set of words that start at $q$, and come back to it, visiting
only lower ranked states than $q$, must be mapped to a unique idempotent $e_q
\in S$.  These properties of $\Aa$ are crafted in such a manner that given any
word $w \in \Sigma^{\infty}$, the unique accepting run of $w$ on $\Aa$ easily
produces a Ramsey split of $w$ (in the sense of Colcombet), the height of the
split being bounded above by the number of states of $\Aa$.

We construct a $\varphi$-good automaton by induction on
$(|S|,|\varphi(\Sigma)|)$ with a lexicographic ordering, a technique introduced
by Wilke \cite{Wilke_1999} and that is very close in spirit to the local divisor
technique of \cite{Diekert_2016}.  The easy base cases of the induction are
when $S$ is a group, and when $|\varphi(\Sigma)|=1$. 
The inductive cases are when we consider a semigroup element $c \in S$ such that
$Sc \subsetneq S$ or $cS \subsetneq S$.  The inductive cases are technically
involved.
The case $Sc \subsetneq S$ is a bit simpler than the other one.  When one deals
with commutative semigroups, we could therefore, simply use this case.  We call
the automaton weakly good if we drop condition (iii) which introduces
non-determinism.  Upto the first inductive case, we can obtain a weakly good
automaton which is deterministic.  But with the second inductive case $cS
\subsetneq S$, things get more complex, and we show that it is not possible to
obtain deterministic weakly good automata.  In a way, the price we pay in
obtaining Ramsey splits is the non-determinism.  This must be contrasted with
the construction of Colcombet \cite{Colcombet_2010}, where a \emph{forward
Ramsey split} is obtained, while retaining determinism in the automaton.  One
way we can avoid non-determinism is to allow look-aheads in the constructed good
automata.  It turns out however that, it is not possible to obtain a bounded
look-ahead, and in general, one needs a regular look-ahead in the constructed
good automaton.  
   
The good automata, though challenging in its construction and proof of
correctness, has some nice take-aways: (1) It provides a very simple proof of
the forest factorization theorem, and (2) it allows us to synthesize good
expressions \cite{lics18} by a standard elimination of states in $\Aa$.  The
properties imposed on $\Aa$ which make it good, helps significantly in both
these cases: (1) in the case of the forest factorization, the states of $\Aa$
are used in labelling the positions of the word; whenever a state repeats, we
declare then equivalent, as long as no higher state has been seen in between.
This trivially gives a Ramsey split, with the height being the number of states
of $\Aa$.  (2) The synthesis of good expressions follows very easily thanks to
the unambiguity of $\Aa$, the ordering on the states, and the condition of
obtaining a unique idempotent while returning to a state without seeing a higher
state.  

Our construction of good automata is in general exponential in the size of the 
semigroup. It would be interesting to study how this construction can be 
optimized.

\section{Unambiguous Forest Factorization}\label{sec:Uforest}
Let $\Sigma$ denote a finite alphabet.  $\Sigma^{\infty}$
represents $\Sigma^* \cup \Sigma^{\omega}$, the set of finite or infinite words. 
Given a word $w=a_1a_2\dots$ with $a_i \in \Sigma$, $w[x, \dots, y]$ denotes 
the word $a_x \dots a_y$. 
For rational expressions over $\Sigma$ we will use the following syntax:
$$
F ::= \emptyset \mid \varepsilon \mid a \mid F\cup F \mid F\cdot F \mid F^+
$$
where $a\in\Sigma$. For reasons that will be clear below, we prefer to use the 
Kleene-plus instead of the Kleene-star, hence we also add $\varepsilon$ 
explicitely in the syntax. An expression is said to be $\varepsilon$-free if it 
does not use $\varepsilon$.
We denote by $\Lang{E}$ the regular language denoted by $E$.

Let $(S,\cdot,\unitS)$ be a \emph{finite} monoid and $\varphi\colon\Sigma^*\to S$ be a
morphism.
We say that a rational expression $F$ is $\varphi$-\emph{good} (or simply
\emph{good} when $\varphi$ is clear from the context) when
\begin{enumerate}[nosep]
  \item the rational expression $F$ is unambiguous,
  
  \item for each subexpression $E$ of $F$ we have 
  $\Lang{E}\subseteq\phi^{-1}(s_E)$ for some $s_E\in S$.
  \item for each subexpression $E^+$ of $F$ we have 
  $\Lang{E}\subseteq\phi^{-1}(s_E)$ for some \emph{idempotent} $s_E\in S$.
\end{enumerate}
\newcommand\rewrites\Rightarrow
Notice that the classical rewrite rules used to simplify 
expressions using $\emptyset$ preserve good expressions. These rewrite rules are
$\emptyset^{+}\rewrites\emptyset$, 
$\emptyset\cdot F\rewrites\emptyset$, 
$F\cdot\emptyset\rewrites\emptyset$, 
$\emptyset\cup F\rewrites F$, 
$F\cup\emptyset\rewrites F$. Hence, each good expression is equivalent to a 
good expression which is either simply $\emptyset$, or does not use 
$\emptyset$ as a subexpression. Also, $\varepsilon$-freeness is preserved by this 
simplification.

\begin{theorem}[Unambiguous Forest Factorization]\label{thm:U-forest}
  Let $\varphi\colon\Sigma^*\to S$ be a morphism to a finite monoid
  $(S,\cdot,\unitS)$.  
  \begin{enumerate}[nosep,label=$(\mathsf{T}_{\arabic*})$,ref=$\mathsf{T}_{\arabic*}$]
  \item\label{item:T1} For each $s\in S$, there is an $\varepsilon$-free
  \emph{good} rational expression $F_s$ such that
    $\Lang{F_s}=\varphi^{-1}(s)\setminus\{\varepsilon\}$.
    Therefore, $G=\bigcup_{s\in S}F_s$ is an unambiguous rational
    expression over $\Sigma$ such that $\Lang{G}=\Sigma^+$.
  
  \item\label{item:T2} There is an \emph{unambiguous} rational expression
  $G=\bigcup_{k=1}^{m}F_k\cdot G_k^\omega$ over $\Sigma$ such that
    $\Lang{G}=\Sigma^\omega$ and for all $1\leq k\leq m$ the expressions $F_k$
    and $G_k$ are $\varepsilon$-free $\varphi$-good rational expressions and
    $s_{k}$ is an idempotent, where $\Lang{G_k}\subseteq\varphi^{-1}(s_{k})$.
  \end{enumerate}
\end{theorem}

The good regular expressions will be obtained using the classical translation of
automata to regular expressions by successive state eliminations.  To this aim,
the automaton should have several properties.  Mainly it should be unambiguous
and there should be a total order on states which is used in the state elimination.
We study these properties in the next section.

%%%%%%%%%%%%%%%%%%%%%%%
\section{Good Automata}\label{sec:good-automata}

Let $\A=(Q,\Sigma,\Delta,\iota,F,R,<)$ be an automaton where $Q$ is the finite
set of states, $\Sigma$ the alphabet, $\Delta\subseteq Q\times\Sigma\times Q$
the transition relation, $\iota\in Q$ is the initial state, $F,R\subseteq Q$ are
the subsets of final and repeated (B\"uchi) states, and $<$ is a total
order on $Q$.  For $p,q\in Q$ and $w\in\Sigma^{+}$ we write $p\xrightarrow{w}q$
when there is a run in $\A$ from $p$ to $q$ reading $w$.  We let $L_{p,q}$ be
the set of nonempty words $w\in\Sigma^{+}$ such that $p\xrightarrow{w}q$.  If
$X\subseteq Q$ then we write $p\xrightarrow{w}_X q$ if there is such a run where
all intermediary states are in $X$.  We let $L_{p,X,q}$ be the set of nonempty
words $w\in\Sigma^{+}$ such that $p\xrightarrow{w}_X q$.  Hence, we have
$L_{p,Q,q}=L_{p,q}$ and $L_{p,\emptyset,q}\subseteq\Sigma$.  We simply write
$L_q=L_{q,\da q,q}\subseteq L_{q,q}$ where $\da q=\{p\in Q\mid p<q\}$.

Let $\phi\colon\Sigma^{+}\to S$ be a semigroup\footnote{We may start from a
monoid morphism but during the induction we will have to consider semigroups.}
morphism.  The automaton $\A$ is $\phi$-good (or simply good) if it satisfies
the following properties:
\begin{enumerate}[nosep,label=$(\mathsf{G}_{\arabic*})$,ref=$\mathsf{G}_{\arabic*}$]
  \item\label{item:G1} $\A$ is unambiguous and universal (accepts all words).
  For each word $w\in \Sigma^{+}\cup\Sigma^{\omega}$ there is one and only one
  accepting run for $w$ in $\A$.
  
  \item\label{item:G2} For all $q\in Q$, there is an idempotent $e_q\in S$ such
  that $L_{q}\subseteq\phi^{-1}(e_q)$, i.e., all words in $L_{q}$ (if any) are
  mapped by $\phi$ to the same semigroup element $e_q$, which is an idempotent.
  
  \item\label{item:G3} The initial state $\iota$ has no 
  incoming transitions and is maximal: $q<\iota$ for all $q\in 
  Q\setminus\{\iota\}$.
  
  \item\label{item:G4} There is only one final state $F=\{f\}$ and $f$ has no 
  outgoing transitions. Moreover, the total order on 
  states satisfies $Q\setminus\{\iota,f\}) < f < \iota$.
\end{enumerate}
We say that $\A$ is weakly-good if it satisfies 
(\ref{item:G1},\ref{item:G2},\ref{item:G3}). 
\begin{lemma}\label{lem:wgood}	
From a \emph{weakly-good} automaton, we can construct an equivalent
\emph{good} automaton.
\end{lemma}
\begin{proof}
  Let $\A=(Q,\Sigma,\Delta,\iota,F,R,<)$ be a weakly-good automaton for the
  morphism $\phi\colon\Sigma^{+}\to S$.  Let $f\notin Q$ be a new state and let
  $Q'=Q\cup\{f\}$.  We define 
  $\A'=(Q',\Sigma,\Delta\cup\Delta',\iota,\{f\},R,<')$ as follows: $\Delta'$ is 
  the set of transitions $(q,a,f)$ such that there is a transition 
  $(q,a,q')\in\Delta$ with $q'\in F$. The ordering $<'$ coincides with $<$ on 
  $Q$ and satisfies $Q\setminus\{\iota\}<'f<'\iota$.
  
  Clearly, \eqref{item:G3} holds for $\A'$.  Notice that $f$ has no outgoing
  transitions, hence \eqref{item:G4} is satisfied.  Also, $L_f(\A')=\emptyset$
  and $L_q(\A')=L_q(\A)$ for all $q\in Q$, hence \eqref{item:G2} holds for
  $\A'$.  Finally, $\A$ and $\A'$ have the same infinite runs and there is a 
  bijection between the finite accepting runs of $\A$ and the finite accepting 
  runs of $\A'$.  We deduce easily that \eqref{item:G1} is satisfied.
\end{proof}

\begin{example}\label{ex:eg-good}
  Consider the morphism $\varphi\colon\Sigma^+ \to S=\{\alpha,\beta\}$ defined by
  $\varphi(a)=\alpha$ and $\varphi(b)=\beta$.  The product in $S$ is so that $\alpha$
  and $\beta$ are both right absorbing ($\alpha s=\alpha$ and $\beta s=\beta$ for all $s\in
  S$) and hence idempotents.  The automaton $\Aa$ (left in Figure~\ref{fig:eg-good})
  is $\varphi$-good.  The ordering on states is $n_b < n_a < f < \iota$,
  $F=\{f\}, R=\{n_a, n_b\}$.  The states $n_a, n_b$ determine the next symbol to
  be read as $a$ and $b$ respectively.  It is easy to see that \ref{item:G1} is
  true: consider a word $w \in \Sigma^{\infty}$.  For all $i>1$, the $i$th
  symbol of $w$ is $x \in \{a,b\}$ iff the $i$th state in the unique accepting
  run from $\iota$ is $n_x$.  By the ordering of states,
  $L_{\iota}=\emptyset=L_f$, $L_{n_b}=\{b\}\subseteq\varphi^{-1}(\beta)$ and
  $L_{n_a}=ab^* \subseteq \varphi^{-1}(\alpha)$.  Since $\alpha$ and $\beta$ are
  idempotents, \ref{item:G2} holds good.  \ref{item:G3} and \ref{item:G4} also
  hold good easily.
  Figure~\ref{fig:eg-good} also depicts on the right $\psi$-good automaton $\Bb$
  for the morphism $\psi\colon\{a,b\}^+ \to
  S=\{\alpha,\beta,\alpha\alpha,\alpha\beta,\beta\alpha,\beta\beta\}$ with
  $\psi(a)=\alpha$, $\psi(b)=\beta$ and the product in $S$ is so that the
  semigroup elements $\alpha\alpha,\alpha\beta,\beta\alpha,\beta\beta$ are right
  absorbant.  The repeated states of $\Bb$ are $R=\{n_{aa}, n_{ab}, n_{bb},
  n_{ba}\}$.  Notice that merging $n_{aa}, n'_{aa}$ (or $n_{bb}, n'_{bb}$)  
  violates \ref{item:G2} (the idempotents in $S$ are
  $\alpha\alpha,\alpha\beta,\beta\alpha,\beta\beta$; the merge will result in $a
  \in L_{n_{aa}}$, but $\psi(a)=\alpha$ is not idempotent.)
\end{example}

 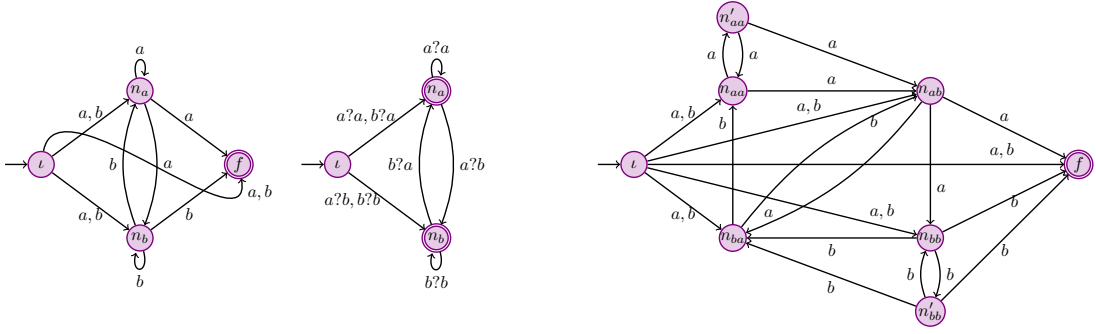
\begin{figure} [t]                                                                                                                                                           \begin{center}                                                                                                                                                         
  \scalebox{0.65}{                                                                                                                                                       
    \begin{tikzpicture}[->,thick]                                                                                                                                          
      \node[initial, cir, initial text ={}] at (-16,0) (A) {$\iota$} ;
      \node[cir] at (-14,1.5) (B) {$n_a$}; 
      \node[cir] at (-14,-1.5) (C) {$n_b$};
      \node[cir,accepting] at (-12,0) (D) {$f$};
      
            \path (A) edge node [above]{$a,b$}node {} (B);
            \path (A) edge node [below]{$a,b$}node  {} (C);
                  \path (B) edge[bend left=20] node [right] {$a$} node  {}(C);
                  \path (C) edge[bend left=20] node [left] {$b$} node  {}(B);
                  
           \path (B) edge[loop above] node[] {$a$} (B);
            \path(C) edge[loop below] node []{$b$} (C);
    \path (B) edge node [above] {$a$} node {}(D);
    \path (C) edge node [below] {$b$} node {}(D);
     \path (A) edge[in=-80,out=80] node [below, at end, anchor=north west]{$a,b$}node [below] {} (D);
    
        \node[initial, cir, initial text ={}] at (-10,0) (A) {$\iota$} ;
      \node[cir,accepting] at (-8,1.5) (B) {$n_a$}; 
      \node[cir,accepting] at (-8,-1.5) (C) {$n_b$};
      
            \path (A) edge node [left,pos=0.7]{$a?a,b?a$}node {} (B);
            \path (A) edge node [left,pos=0.5]{$a?b,b?b$}node  {} (C);
                  \path (B) edge[bend left=20] node [right] {$a?b$} node  {}(C);
                  \path (C) edge[bend left=20] node [left] {$b?a$} node  {}(B);
                  
           \path (B) edge[loop above] node[] {$a?a$} (B);
            \path(C) edge[loop below] node []{$b?b$} (C);

      \node[initial, cir, initial text ={}] at (-4,0) (A1) {$\iota$} ;
      \node[cir] at (-2,3) (B11) {$n'_{aa}$}; 
            \node[cir] at (-2,1.5) (B1) {$n_{aa}$}; 
      \node[cir] at (-2,-1.5) (C1) {$n_{ba}$};
      \node[cir] at (2,1.5) (B2) {$n_{ab}$}; 
      \node[cir] at (2,-1.5) (C2) {$n_{bb}$};
      \node[cir] at (2,-3) (C22) {$n'_{bb}$};
             \node[cir,accepting] at (5,0) (D1) {$f$};
      
           \path (A1) edge node [above]{$a,b$}node {} (B1);
            \path (A1) edge node [below]{$a,b$}node  {} (C1);
             \path (A1) edge node [above,pos=0.6]{$a,b$}node {} (B2);
            \path (A1) edge node [above,pos=0.87]{$a,b$}node  {} (C2);
            \path (A1) edge node [above,pos=0.85]{$a,b$}node {} (D1);
              \path (B1) edge node [above] {$a$} node  {}(B2);
           \path (B1) edge[bend left=20] node[left] {$a$} (B11);
           \path (B11) edge[bend left=20] node[right] {$a$} (B1);
           \path (B11) edge node[above] {$a$} (B2);

            \path (B2) edge node[right,pos=0.7]   {$a$} node  {}(C2);
            \path (B2) edge[bend left=15] node [above,pos=0.9] {$a$} node  {}(C1);
           \path (C2) edge[bend left=20] node[right] {$b$} (C22);
           \path (C22) edge[bend left=20] node[left] {$b$} (C2);
           \path (C22) edge node[below] {$b$} (C1);
           \path (C22) edge node[below] {$b$} (D1);

            \path (C2) edge node[below] {$b$} node  {}(C1);
            \path (C1) edge[bend left=15] node [right,pos=0.75] {$b$} node  {}(B2);
            \path (C1) edge node [left,pos=0.85] {$b$} node  {}(B1);
            \path (B2) edge node [above] {$a$} node  {}(D1);
           \path (C2) edge node[right] {$b$} (D1);

    \end{tikzpicture}
  }
  \caption{On the left, a $\varphi$-good automaton $\Aa$ for $\varphi\colon
  \{a,b\}^+ \to \{\alpha,\beta\}$, $\varphi(a)=\alpha$, $\varphi(b)=\beta$ and
  $xy=x$ for all $x,y\in S$.  In the middle, a weakly-good automaton for $\phi$ 
  which is deterministic and complete with one letter look-ahead. Here, the 
  label $x?y$ means reading $x$ with look-ahead $y$. On the right,
  is automaton $\Bb$ which is $\psi$-good for the morphism $\psi\colon\{a,b\}^+
  \to \{\alpha,\beta,\alpha\alpha,\alpha\beta,\beta\alpha,\beta\beta\}$ with
  $\psi(a)=\alpha$, $\psi(b)=\beta$ and $xyz=xy$ for all
  $x,y,z\in\{\alpha,\beta\}$. In all the figures, we use double circle to denote final states.
   }                                                                                                                                         
   \label{fig:eg-good}                                                                                                                                                        
\end{center}                                                                                                                                                         
\end{figure}

We now move towards the main result.  Let 
 ${\phi\colon\Sigma^{+}\to S}$ be a semigroup morphism. 
\begin{theorem}\label{thm:good-automaton}
  Given $\varphi$ as above,  we can construct a 
  $\phi$-\emph{good} automaton $\A_\phi$.
\end{theorem}

The proof is by induction on $(|S|,|\phi(\Sigma)|)$ with lexicographic ordering. 
Wilke \cite{Wilke_1999} used this kind of induction while obtaining a temporal
logic formula from a counter-free $\omega$-automata.
The survey of Kufleitner and Diekert on local divisor technique
\cite{Diekert_2016} uses a similar induction to prove Simon's Forest
factorisation theorem. See also the survey \cite{DiGa08Thomas} where the local 
divisor technique was used to obtain an LTL formula from an aperiodic monoid.

\subsection*{Base Cases}
A first basic case is when $S$ is a group, which is in particular the case when 
$|S|=1$.

\begin{lemma}\label{lem:group}
  If $S$ is a group, we can construct a deterministic and complete
  \emph{weakly-good} automaton for the morphism $\phi\colon\Sigma^{+}\to S$.
\end{lemma}

\begin{proof}
  We let $Q=S\uplus\{\iota\}$.  The initial state is $\iota$.  All other states
  are accepting: $F=R=S$.  The deterministic transition function is defined by
  $\iota\xrightarrow{a}\phi(a)$ and $s\xrightarrow{a}s\cdot\phi(a)$ for all
  $s\in S$ and $a\in\Sigma$.  \eqref{item:G1} holds trivially since the
  automaton is deterministic and complete.  Now, for $s,t\in S$, we check that
  $L_{s,t}=\phi^{-1}(s^{-1}t)$.  In particular, if $s=t$ then
  $L_{s,s}=\phi^{-1}(1_S)$ where $1_S$ is the unit of $S$ which is indeed
  idempotent.  We deduce that \eqref{item:G2} is also satisfied whichever total
  order $<$ is chosen on $Q$. We assume $S<\iota$ so that \eqref{item:G3} is 
  also satisfied.  
\end{proof}

\begin{example}
	As an example illustrating Lemma \ref{lem:group}, consider the morphism
	$\varphi\colon \{a,b\}^{+} \to S$, where $S=((\mathbb{Z}/2\mathbb{Z})^2,+)$,
	the group of pairs $(x,y) \in \{0,1\}^2$ with component wise
	addition, defined by $\varphi(a)=(1,0)$ and $\varphi(b)=(0,1)$.  (0,0) is
	the unit element.  The $\varphi$-weakly good automaton $\mathcal{A}$ is given
	in Figure~\ref{fig:group}.  Since the automaton is deterministic and complete,
	\eqref{item:G1} is easy.  To see \eqref{item:G2}, observe that for any state
	$q$, $L_{q,q}$ is the set of all words with even number of $a$s and $b$s, and
	indeed, $\varphi(L_{q,q})$ is (0,0), the unit element.  This shows that $L_q
	\subseteq L_{q,q} \subseteq \varphi^{(-1)}((0,0))$, satisfying
	\eqref{item:G2}.  Finally, \eqref{item:G3} holds trivially by construction on
	choosing an ordering of states respecting $q < \iota$ for all $q \neq \iota$.
\end{example}

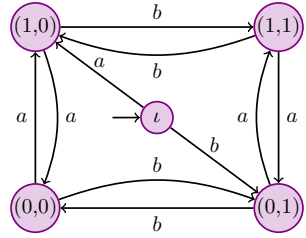
\begin{figure} [t]                                                                                                                                                           \begin{center}                                                                                                                                                         
  \scalebox{0.8}{                                                                                                                                                       
    \begin{tikzpicture}[->,thick]                                                                                                                                          
      \node[initial, cir, initial text ={}] at (-12,-1.5) (A) {$\iota$} ;
      \node[cir] at (-14,0) (B) {(1,0)}; 
      \node[cir] at (-14,-3) (C) {(0,0)};
      \node[cir] at (-10,0) (D) {(1,1)};
      \node[cir] at (-10,-3) (E) {(0,1)};
      
            \path (A) edge node [above]{$a$}node [below] {} (B);
            \path (A) edge node [above]{$b$}node [below] {} (E);
            
      \path (B) edge[bend left=20] node [right] {$a$} node {}(C);
      \path (B) edge node [above] {$b$} node  {}(D);
      
    \path (C) edge node [left] {$a$} node {}(B);
      \path (C) edge[bend left=20] node [above] {$b$} node  {}(E);

    \path (D) edge node [right] {$a$} node [right] {}(E);
      \path (D) edge[bend left=20] node [below] {$b$} node {}(B);
      
      \path (E) edge node [below] {$b$} node [above] {}(C);
      \path (E) edge[bend left=20] node [left] {$a$} node{}(D);

    \end{tikzpicture}
  }
  \caption{The $\varphi$-weakly good automaton $\mathcal{A}$, where 
  $\varphi\colon \{a,b\}^{+} \to S$,  $S$ is the group $((\mathbb{Z}/2\mathbb{Z})^2,+)$. 
  It is easy to see that for all states $q$, $L_{q,q} =\varphi^{-1}(0,0)$, and (0,0) is the unit. 
  }                                                                                                                                       
  \label{fig:group}                                                                                                                                                        
\end{center}                                                                                                                                                         
\end{figure} 
The second basic case is when all letters from $\Sigma$ are mapped to the same 
semigroup element, i.e., when $|\phi(\Sigma)|=1$.

\begin{lemma}\label{lem:one-generator}
  If all letters are mapped to the same semigroup element $s$, i.e.,
  $\phi(\Sigma)=\{s\}$, we can construct a deterministic and complete
  \emph{weakly-good} automaton for the morphism $\phi\colon\Sigma^{+}\to S$.
\end{lemma}

\begin{proof}
  Since $S$ is finite, there are integers $k,\ell\geq 1$ such that 
  $s^{k}=s^{k+\ell}$. We fix the least such pair for the lexicographic order.
  Also, since $S$ is finite, we find $n\geq 1$ such that $s^{n}$ is an 
  idempotent. Again, we fix the least such $n$. It is easy to see that $k \leq n$. 
  Also, $n<k+\ell$ by minimality of $n$, since otherwise we have 
  $s^{n}=s^{n-\ell}$.
  Further, from $s^{n}=s^{n+n}$ we deduce that $\ell$ divides $n$.
  
  Now, we define the automaton.  The set of states is $Q=\{0,1,\ldots,k+n-1\}$ and
  the initial state is $\iota=0$.  All states are accepting: $F=R=Q$.
  The deterministic and complete transition function is defined as expected: for
  all $a\in\Sigma$ and $i\in Q$ we let $i\xrightarrow{a}j$ where $j=i+1$ if
  $i<k+n-1$ and $j=k$ otherwise (see Figure~\ref{fig:all-one}). 
   
  \eqref{item:G1} holds trivially since the automaton is deterministic and
  complete.  Now, $L_{i,i}\neq\emptyset$ if and only if $k\leq i\leq k+n-1$ and
  in this case $L_{i,i}=(\Sigma^{n})^{+}\subseteq\phi^{-1}(s^{n})$ since $s^{n}$
  is an idempotent.  Therefore, \eqref{item:G2} holds. For \eqref{item:G3} to 
  hold, we simply have to consider a total order where $\iota$ is maximal.
\end{proof}

\begin{example}
  As an example illustrating Lemma~\ref{lem:one-generator}, consider the
  morphism $\varphi\colon \{a,b\}^{+} \to S$ defined by
  $\varphi(a)=\varphi(b)=s$ where $S$ is the finite semigroup $S=\{s,s^2, s^3,
  s^4\}$ with $s^5=s^3$.  We have $k=3$ and $\ell=2$ and the idempotent
  is $s^n=s^4$.  The automaton $\mathcal{A}$ for $\varphi$ is in
  Figure~\ref{fig:all-one}.
\end{example}

\begin{figure} [t]                                                                                                                                                           \begin{center}                                                                                                                                                         
  \scalebox{0.8}{                                                                                                                                                       
    \begin{tikzpicture}[->,thick]                                                                                                                                          
      \node[initial, cir1, initial text ={}] at (-12,-2) (A) {0} ;
      \node[cir1] at (-11,-2) (B) {1}; 
      \node[cir1] at (-10,-2) (C) {2};
      \node[cir1] at (-9,-2) (D) {3};
      \node[cir1] at (-8,-2) (E) {4};
      \node[cir1] at (-7,-2) (F) {5};
      \node[cir1] at (-6,-2) (G) {6};
       \path (A) edge node [above]{$\Sigma$}node {}(B);
       \path (B) edge node [above]{$\Sigma$}node {}(C);
        \path (C) edge node [above]{$\Sigma$}node{} (D);
       \path (D) edge node [above]{$\Sigma$}node {}(E);
       \path (E) edge node [above]{$\Sigma$}node{} (F);
       \path (F) edge node [above]{$\Sigma$}node {}(G);
       \path (G) edge[bend left=40] node [above]{$\Sigma$}node{} (D);
    \end{tikzpicture}
  }
  \caption{The $\varphi$-weakly good automaton $\mathcal{A}$ where 
  $\varphi\colon \{a,b\}^{+} \to S$, $S=\{s,s^2,s^3,s^4\}$, and 
  $\varphi(a)=\phi(b)=s$. All states are accepting.}                                                                                                                                       
  \label{fig:all-one}                                                                                                                                                        
\end{center}                                                                                                                                                         
\end{figure}
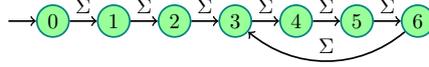 

\subsection*{Inductive Steps}
The other two cases are inductive.  First, assume that there is some semigroup
element $c\in\phi(\Sigma)$ such that $Sc \subsetneq S$.  Then $(Sc,\cdot)$ is a
strict subsemigroup\footnote{Notice that, if
$S$ is a monoid with unit $1_S$ then $1_S\notin Sc$ (otherwise $Sc=S$).  Hence
$Sc$ is not a submonoid of $S$.  Moreover, $Sc$ may not have a unit element.
This is why we consider semigroup morphisms.  Another possibility would be the
local divisor technique described in \cite{Diekert_2016} which allows to get a
smaller monoid.}  of $(S,\cdot)$, i.e., $|Sc|<|S|$.
Let $\Sigma_2=\Sigma\cap\phi^{-1}(c)$ be the set of all letters
mapped to $c$ and $\Sigma_1=\Sigma\setminus\Sigma_2$.  If $\Sigma_1=\emptyset$
then we are in the second basic case above.  Hence we assume
$\Sigma_1\neq\emptyset$ and since $c\in\phi(\Sigma)\setminus\phi(\Sigma_1)$
we have $|\phi(\Sigma_1)|<|\phi(\Sigma)|$ so by induction hypothesis
we can construct a \emph{good} automaton
$\A_1=(Q_1,\Sigma_1,\Delta_1,\iota_1,F_1,R_1,<_1)$ for the morphism $\phi$
restricted to $\Sigma_1$.
Each nonempty word $w\in\Sigma^{+}\cup\Sigma^{\omega}$ has a unique 
factorization 
$w=(a_1 u_1 c_1)(a_2 u_2 c_2)(a_3 u_3 c_3)\cdots$
with $a_i\in\Sigma$, $u_i\in\Sigma_1^{*}\cup\Sigma_1^{\omega}$ and
$c_i\in\Sigma_2$.  If the word $w\in\Sigma^{+}$ is finite then the factorization
has $n\geq 1$ blocks, the last block being either $a_nu_n$ or $a_nu_nc_n$.  If
$w\in\Sigma^{\omega}$ is infinite, the factorization has infinitely many
blocks when $w$ has infinitely many letters from $\Sigma_2$, otherwise the
factorization has $n\geq 1$ blocks and the last block is $a_nu_n$ with
$u_n\in\Sigma_1^{\omega}$.

We view $B=\phi(\Sigma\Sigma_1^{*}\Sigma_2)\subseteq Sc$ as an alphabet and we
consider the evaluation semigroup morphism $\psi\colon B^{+}\to Sc$ defined by 
$\psi(b)=b$ for all $b\in B\subseteq Sc$. Let 
$b_i=\phi(a_iu_ic_i)\in B$. The factorization of $w$ yields the word 
$b_1b_2b_3\cdots$ over $B$. Moreover, for $i\leq j$ we have $\psi(b_i\cdots 
b_j)=\phi(a_iu_ic_i\cdots a_ju_jc_j)$.
Since $|Sc|<|S|$, we can construct a \emph{good} automaton
$\A_2=(Q_2,B,\Delta_2,\iota_2,F_2,R_2,<_2)$ for the morphism $\psi\colon
B^{+}\to Sc$ by induction hypothesis.

\begin{example}
  We give an example illustrating the first inductive case $Sc \subsetneq S$.
  Consider the finite semigroup $S=\{s,s^2,s^3,s^4\}$ with $s^5=s^3$ and $s^4$
  idempotent.  Consider the morphism $\varphi\colon\{a,b\}^{+} \to S$ with
  $\varphi(a)=s$, $\varphi(b)=s^2$.  Choosing $c=s^2$, we see that $Ss^2=\{s^3,
  s^4\} \subsetneq S$.  It can be seen that $Ss^2$ is a group with unit element
  $s^4$.  Considering $\Sigma_2=\{b\}$ and $\Sigma_1=\{a\}$ we have
  $\varphi(\Sigma_2)=s^2$, and $s^2 \notin \varphi(\Sigma_1)$.  The inductive
  hypothesis applies to $\varphi_1 \colon \Sigma_1^{+} \to S$; since $\Sigma_1$
  is a singleton, the $\varphi_1$-good automaton $\mathcal{A}_1$ can be
  constructed as in Lemma~\ref{lem:one-generator} (see
  Figure~\ref{fig:induct-one}).  Also, considering $B=\varphi(\Sigma \Sigma_1^*
  \Sigma_2)=\varphi(a^+b) \cup \varphi(ba^*b)=\{s^3,s^4\}=Ss^2$, we have the
  evaluation morphism $\psi\colon B^+ \to Ss^2$, given by $\psi(s^{3})=s^{3}$
  and $\psi(s^{4})=s^{4}$.  The inductive hypothesis applies to $\psi$: in
  particular, the $\psi$-good automaton $\mathcal{A}_2$ can be constructed as in
  Lemma~\ref{lem:group} (see Figure~\ref{fig:induct-one}).

  We explain below how to construct a weakly-good automaton $\Aa$ for $\phi$
  from $\Aa_1$ and $\Aa_2$.  Consider the word $w=(aab)(baaab)(ab) \in (\Sigma
  \Sigma_1^* \Sigma_2)^+$.  Let $b_1=s^3$, $b_2=s^4$.  Then
  $\varphi(w)=b_2b_1b_1 \in B^+$.  Figure~\ref{fig:induct-one} depicts some
  example runs of $\Aa$.
%   \notekr{Please check.  The runs being different even though the
%   illustrated words are subwords growing larger should highlight the nature of
%   the runs.}
\end{example}

\begin{figure} [t]                                                                                                                                                           
  \begin{center}                                                                                                                                                         
  \scalebox{0.8}{                                                                                                                                                       
    \begin{tikzpicture}[->,thick]                                                                                                                                          
    
      \node[initial, cir1, initial text ={}] at (-12,-2) (A) {0} ;
      \node[cir1] at (-11,-2) (B) {1}; 
      \node[cir1] at (-10,-2) (C) {2};
      \node[cir1] at (-9,-2) (D) {3};
      \node[cir1] at (-8,-2) (E) {4};
      \node[cir1] at (-7,-2) (F) {5};
      \node[cir1] at (-6,-2) (G) {6};
      \node[cir1,accepting] at (-9,0) (F1) {$f_1$};
      
       \path (A) edge node [above]{$a$}node {}(B);
       \path (A) edge[bend left=40] node [above,pos=0.9]{$a$}node {}(F1);
       
       \path (B) edge node [above]{$a$}node {}(C);
       \path (B) edge[bend left=20] node [above,pos=0.8]{$a$}node {}(F1);
       
        \path (C) edge node [above]{$a$}node{} (D);
        \path (C) edge[bend left=10] node [above,pos=0.7]{$a$}node {}(F1);
       
       \path (D) edge node [above]{$a$}node {}(E);
       \path (D) edge node [left,pos=0.6]{$a$}node {}(F1);
       
       \path (E) edge node [above]{$a$}node{} (F);
       \path (E) edge[bend right=10] node [right,pos=0.9]{$a$}node {}(F1);
       
       \path (F) edge node [above]{$a$}node {}(G);
\path (F) edge[bend right=20] node [right,pos=0.8]{$a$}node {}(F1);
              \path (G) edge[bend left=40] node [above]{$a$}node{} (D);
              \path (G) edge[bend right=30] node [above,pos=0.7]{$a$}node {}(F1);

      \node[initial, cir2, initial text ={}] at (-4,-2) (A1) {$\iota_2$} ;
        \node[cir2,accepting] at (-4,0) (F2) {$f_2$} ;
        \node[cir2] at (-1,-2) (B1) {$b_1$} ;
        \node[cir2] at (-1,0) (C1) {$b_2$} ;

       \path (A1) edge node [below]{$b_1$}node {}(B1);
       \path (A1) edge node [above,pos=0.8]{$b_2$}node {}(C1);
       \path (A1) edge node [left]{$b_2,b_1$}node {}(F2);
       \path (B1) edge node [below,pos=0.8]{$b_2,b_1$}node {}(F2);

       \path (B1) edge[loop below] node [right]{$b_2$}node {}(B1);
       \path (B1) edge[bend left=20] node [left]{$b_1$}node {}(C1);

       \path (C1) edge[loop above] node [right]{$b_2$}node {}(C1);
       \path (C1) edge[bend left=20] node [right]{$b_1$} node{} (B1);
       \path (C1) edge node [above]{$b_1,b_2$} node{} (F2);

  \node[initial, cir2, initial text ={}] at (-14,-4) (A2) {$\iota_2$} ;
  \node[cir11,accepting] at (-12.5,-4) (B2) {$(\iota_2,s,0)$} ; 
  \path(A2) edge node [above]{$a$}node {}(B2);
  \node[cir11,accepting] at (-10.5,-4) (C2) {$(\iota_2,s^2,f_1)$} ;
  \path(B2) edge node [above]{$a$}node {}(C2);
  \node[cir2] at (-9,-4) (D2) {$b_2$} ;
  \path(C2) edge node [above]{$b$}node {}(D2);
  \node[cir11] at (-7.5,-4) (E2) {$(b_2,s^2,0)$} ;
  \path(D2) edge node [above]{$b$}node {}(E2);
  \node[cir11] at (-5.5,-4) (F2) {$(b_2,s^3,1)$} ;
  \path(E2) edge node [above]{$a$}node {}(F2);
  \node[cir11] at (-3.5,-4) (G2) {$(b_2,s^4,2)$} ;
  \path(F2) edge node [above]{$a$}node {}(G2);
  \node[cir11] at (-1.5,-4) (H2) {$(b_2,s^3,f_1)$} ;
  \path(G2) edge node [above]{$a$}node {}(H2);
  \node[cir2] at (0,-4) (I2) {$b_1$} ;
  \path(H2) edge node [above]{$b$}node {}(I2);
  \node[cir11] at (1.5,-4) (J2) {$(b_1,s,0)$} ;
  \path(I2) edge node [above]{$a$}node {}(J2);
  \node[cir2,accepting] at (3,-4) (K2) {$f_2$} ;
  \path(J2) edge node [above]{$b$}node {}(K2);

  \node[initial, cir2, initial text ={}] at (-14,-5) (A2) {$\iota_2$} ;
  \node[cir11,accepting] at (-12.5,-5) (B2) {$(\iota_2,s^2,0)$} ; 
\path(A2) edge node [above]{$b$}node {}(B2);
  \node[cir2,accepting] at (-11,-5) (B3) {$f_2$} ; 
\path(B2) edge node [above]{$b$}node {}(B3);
 \node[cir11,accepting] at (-9.5,-5) (B4) {$(f_2,s^2,0)$} ; 
\path(B3) edge node [above]{$b$}node {}(B4);
 \node[cir11,accepting] at (-7.5,-5) (B5) {$(f_2,s^3,f_1)$} ; 
\path(B4) edge node [above]{$a$}node {}(B5);

  \node[initial, cir2, initial text ={}] at (-14,-6) (A2) {$\iota_2$} ;
  \node[cir11,accepting] at (-12.5,-6) (B2) {$(\iota_2,s^2,0)$} ; 
\path(A2) edge node [above]{$b$}node {}(B2);
  \node[cir2] at (-11,-6) (B3) {$b_2$} ; 
\path(B2) edge node [above]{$b$}node {}(B3);
 \node[cir11] at (-9.5,-6) (B4) {$(b_2,s^2,0)$} ; 
\path(B3) edge node [above]{$b$}node {}(B4);
 \node[cir11] at (-7.5,-6) (B5) {$(b_2,s^3,f_1)$} ; 
\path(B4) edge node [above]{$a$}node {}(B5);
  \node[cir2,accepting] at (-6,-6) (B6) {$f_2$} ;
\path(B5) edge node [above]{$b$}node {}(B6);

  \node[initial, cir2, initial text ={}] at (-14,-7) (A2) {$\iota_2$} ;
  \node[cir11,accepting] at (-12.5,-7) (B2) {$(\iota_2,s^2,0)$} ; 
\path(A2) edge node [above]{$b$}node {}(B2);
  \node[cir2] at (-11,-7) (B3) {$b_2$} ; 
\path(B2) edge node [above]{$b$}node {}(B3);
 \node[cir11] at (-9.5,-7) (B4) {$(b_2,s^2,0)$} ; 
\path(B3) edge node [above]{$b$}node {}(B4);
 \node[cir11] at (-7.5,-7) (B5) {$(b_2,s^3,f_1)$} ; 
\path(B4) edge node [above]{$a$}node {}(B5);
  \node[cir2] at (-6,-7) (B6) {$b_1$} ;
\path(B5) edge node [above]{$b$}node {}(B6);
  \node[cir11] at (-4.5,-7) (B7) {$(b_1,s^2,0)$} ; 
  \node[cir2,accepting] at (-3,-7) (B8) {$f_2$} ; 
  \node[cir11,accepting] at (-1.5,-7) (B9) {$(f_2,s,0)$} ; 
\path(B6) edge node [above]{$b$}node {}(B7);
\path(B7) edge node [above]{$b$}node {}(B8);
\path(B8) edge node [above]{$a$}node {}(B9);

\end{tikzpicture}
  }
  \caption{Top left, is the $\varphi_1$-good automaton $\mathcal{A}_1$, and top
  right, is the $\psi$-good automaton $\mathcal{A}_2$. $b_1=s^3$, $b_2=s^4$.
    Below, a run of the automaton
  $\mathcal{A}$ on $w=(aab)(baaab)(ab)$ is depicted.  The green states summarize
  the run in $\mathcal{A}_1$ on $\Sigma_1^*$, and between consecutive pink
  states, a run in $\mathcal{A}_2$ on $B$ is summarized.
  Runs on $(bb)(ba)$, $(bb)(bab)$ and $(bb)(bab)(bb)a$ are also shown.   
  }
  \label{fig:induct-one}                                                                                                                                                        
\end{center}                                                                                                                                                         
\end{figure}
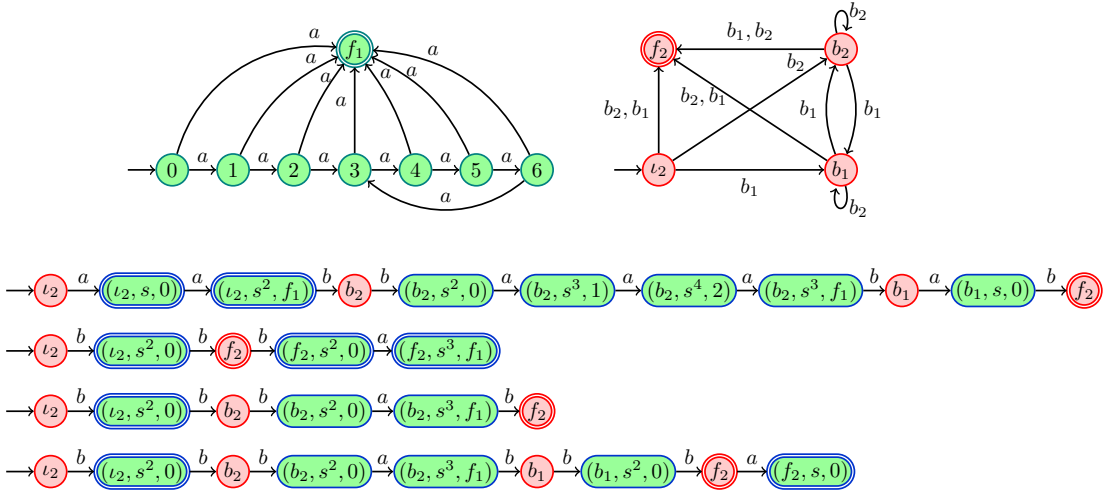 

We now show how to construct a \emph{weakly-good}
automaton $\A$ for $\phi\colon\Sigma^{+}\to S$.  Intuitively, we use $\A_1$ to
scan the words $u_i$ over $\Sigma_1$ and we use $\A_2$ to scan the sequence of
blocks $a_iu_ic_i$ represented by the letters $b_i$ in $B$. 
The set of states of $\A$ is $Q=Q_2\cup(Q_2\times S\times Q_1)$. The 
initial state is $\iota=\iota_2$. 
The transitions are defined below in 
such a way that:
\begin{enumerate}[nosep]
  \item  If 
  $\iota_2\xrightarrow{b_1}q_1\xrightarrow{b_2}q_2\xrightarrow{b_3}q_3\cdots$ 
  is a run of $\A_2$ then we will have in $\A$ the run
  $$
  \iota\xrightarrow{a_1u_1c_1}q_1\xrightarrow{a_2u_2c_2}q_2
  \xrightarrow{a_3u_3c_3}q_3\cdots
  $$

  \item Now, zooming in some factor $a_iu_ic_i$ with $u_i=d_1d_2\cdots d_m$, if
  $\iota_1\xrightarrow{d_1}p_1\xrightarrow{d_2}p_2\cdots\xrightarrow{d_m}p_m$ 
  is a run of $\A_1$ then, with $q=q_{i-1}$, we will have in $\A$ the run
  $$
  q\xrightarrow{a_i}(q,\phi(a_i),\iota_1)
  \xrightarrow{d_1}(q,\phi(a_id_1),p_1)
  \xrightarrow{d_2}(q,\phi(a_id_1d_2),p_2) \cdots
  \xrightarrow{d_m}(q,\phi(a_iu_i),p_m) 
  \xrightarrow{c_i}q_i
  $$
\end{enumerate}
Formally, the transitions of $\A$ are defined as follows:
\begin{itemize}[nosep]
  \item  $q \xrightarrow{a\in\Sigma} (q,\phi(a),\iota_1)$ for $q\in Q_2$,

  \item  $(q,s,p) \xrightarrow{a\in\Sigma_1} (q,s\phi(a),p')$ if 
  $p\xrightarrow{a}p'$ in $\A_1$,

  \item  $(q,s,p) \xrightarrow{a\in\Sigma_2} q'$ if $p\in 
  F_1\cup\{\iota_1\}$ and $q\xrightarrow{sc}q'$ in $\A_2$.
\end{itemize}
Notice that if $\A_1$ and $\A_2$ are deterministic and 
complete then so is the automaton $\A$.

\begin{figure}[t]
  \centering
  \includegraphics[scale=0.22]{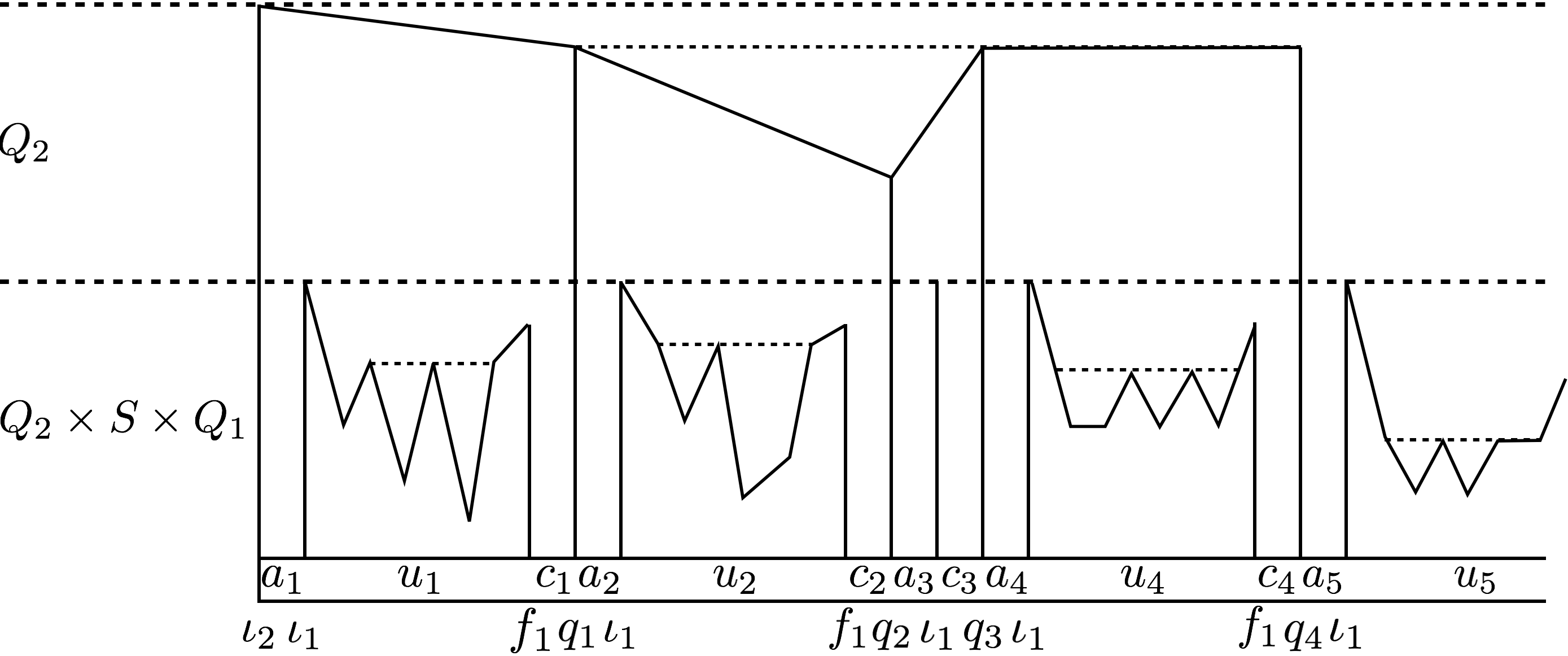}
  \caption{Run for the first inductive case: $Sc\subsetneq S$.}
  \label{fig:RunInductive1}
\end{figure}

The total order $<$ on $Q$ is defined so that $Q_2\times S\times Q_1 < Q_2$, and
$<$ coincides with $<_2$ on $Q_2$, and $p<_1 p'$ implies $(q,s,p)<(q,s',p')$.
Notice that the initial state $\iota=\iota_2$ is the maximal state in $Q$ and
has no incoming transitions, so \eqref{item:G3} holds.
Figure~\ref{fig:RunInductive1} describes the ordering.  While summarizing the
runs of $\Aa_1$ on $u_i \in \Sigma_1^+$, all states are ranked strictly lower
than the states of $\Aa_2$; hence, between two consecutive visits to $Q_2$, all
states seen are strictly lower.  Intuitively this suggests that $\phi(L_q(\Aa))$
for $q \in Q_2$ is same as $\psi(L_q(\Aa_2))$.  Likewise, while staying in
$Q_2\times S\times Q_1$, the ordering of states is that in $\Aa_1$.  Hence,
while considering $L_{(q,s,p)}(\Aa)$, we cannot see any $r \in Q_2$ in the loop;
hence, $\varphi(L_{(q,s,p)}(\Aa))$ must be same as $\varphi(L_p(\Aa_1))$.  This
ensures \eqref{item:G2}.  
The final and repeated states of $\A$ are given by
$F = F_2 \cup \big( (F_2\cup\{\iota_2\})\times S\times(F_1\cup\{\iota_1\}) \big)$, 
$R = R_2 \cup \big( (F_2\cup\{\iota_2\})\times S\times R_1 \big)$.

\begin{lemma}\label{lem:case1-good}
  The automaton $\A$ defined above is \emph{weakly-good} for $\phi\colon\Sigma^{+}\to S$.
\end{lemma}

\begin{proof}
  We have already seen that $\A$ satisfies \eqref{item:G3}.
  We show that $\Aa$ satisfies \eqref{item:G1}.
   
  Consider a word $w\in\Sigma^{+}\cup\Sigma^{\omega}$ and its
  unique factorization $w=(a_1 u_1 c_1)(a_2 u_2 c_2)(a_3 u_3 c_3)\cdots$ with
  $a_i\in\Sigma$, $u_i\in\Sigma_1^{*}\cup\Sigma_1^{\omega}$ and
  $c_i\in\Sigma_2$.  Let $b_i=\phi(a_iu_ic_i)\in B$.  There is a unique empty or
  accepting run $\tau=\iota_2\xrightarrow{b_1}q_1\xrightarrow{b_2}q_2
  \xrightarrow{b_3}q_3\cdots$ of $\A_2$. For each $i\geq 1$, assuming that 
  $u_i=d_1d_2\cdots$, there is a unique empty or accepting run 
  $\sigma_i=\iota_1\xrightarrow{d_1}p_1\xrightarrow{d_2}p_2\cdots$ of $\A_1$. We 
  construct the corresponding  
  subrun 
  $
  \rho_i=q_{i-1}\xrightarrow{a_i}(q_{i-1},\phi(a_i),\iota_1)
  \xrightarrow{d_1}(q_{i-1},\phi(a_id_1),p_1)
  \xrightarrow{d_2}(q_{i-1},\phi(a_id_1d_2),p_2) \cdots
  $ 
  of $\A$.  If $u_i$ is finite with length $m\geq0$ then the last state of
  $\rho_i$ is $(q_{i-1},\phi(a_iu_i),p_m)$ with $p_m\in F_1\cup\{\iota_1\}$ 
  (we let $p_0=\iota_1$).
  In this case, if $c_i$ exists there is a transition
  $(q_{i-1},\phi(a_iu_i),p_m)\xrightarrow{c_i}q_i$ in $\A$ since
  $b_i=\phi(a_iu_i)c$ and $q_{i-1}\xrightarrow{b_i}q_i$ is a transition of
  $\A_2$.  Therefore, $\rho_{i-1}\xrightarrow{c_i}q_i$ is a subrun of $\A$ 
  reading $a_iu_ic_i$.
  
  When $w$ contains infinitely many letters from $\Sigma_2$, the factorization
  is infinite and we obtain the run
  $\rho=\rho_1\xrightarrow{c_1}\rho_2\xrightarrow{c_2}\rho_3\cdots$ of $\A$ 
  reading $w$. Since $\tau$ is accepting in $\A_2$ we have $q_i\in R_2$ 
  for infinitely many $i$'s. Therefore, $\rho$ is accepting in $\A$.
  
  Assume now that $w$ contains finitely many letters from $\Sigma_2$.  Then the
  factorization is finite, say of length $n>0$.  If the last factor $a_nu_nc_n$
  is complete then $\rho=\rho_1\xrightarrow{c_1}\rho_2\xrightarrow{c_2}\cdots
  \rho_{n}\xrightarrow{c_n}q_n$ is a run of $\A$ reading $w$ which is accepting
  since $\tau$ is accepting.
  
  When the last factor is of the form $a_nu_n$ then
  $\rho=\rho_1\xrightarrow{c_1}\rho_2\xrightarrow{c_2}\cdots
  \rho_{n-1}\xrightarrow{c_{n-1}}\rho_n$ is a run of $\A$ reading $w$.  Since
  $\tau$ is empty or accepting, we have $q_{n-1}\in F_2\cup\{\iota_2\}$.  Since
  $\sigma_n$ is an empty or an accepting finite or infinite run of $\A_1$, we
  deduce that $\A$ is accepting.  We have proved that the automaton $\A$ accepts
  all words in $\Sigma^{+}\cup\Sigma^{\omega}$.
  
  \medskip
  We show now that $\A$ is unambiguous.  Let $\rho'$ be an accepting run of $\A$
  on $w$.  We have to show that $\rho'=\rho$ where $\rho$ is the accepting run
  for $w$ defined above.  By definition of $\A$, the run $\rho'$ induces the
  very same factorization of $w=(a_1 u_1 c_1)(a_2 u_2 c_2)\cdots$.  Moreover, we
  can write
  $\rho'=\iota\xrightarrow{a_1u_1c_1}q'_1\xrightarrow{a_2u_2c_2}q'_2\cdots$ and
  $\tau'=\iota_2\xrightarrow{b_1}q'_1\xrightarrow{b_2}q'_2\cdots$ is a run
  of $\A_2$. We show that $\tau'=\tau$.
  
  If $w$ has infinitely many letters from $\Sigma_2$ then the run $\tau'$ is
  infinite and none of the states $q'_i$ belongs to $F_2\cup\{\iota_2\}$ since
  $\A_2$ is good.  Now $\rho'$ is accepting in $\A$ and by definition of $R$ we
  deduce that $q'_i\in R_2$ for infinitely many $i$'s.  Therefore $\tau'$ is
  accepting in $\A_2$.  Since $\A_2$ satisfies \eqref{item:G1}, we deduce that
  $\tau'=\tau$, i.e., $q'_i=q_i$ for all $i$.
  
  If $w$ has finitely many letters from $\Sigma_2$ and the last
  factor is of the form $a_nu_nc_n$ then $\rho'$ ends in state $q'_n\in F$. We 
  deduce that $q'_n\in F_2$ and $\tau'$ is accepting in $\A_2$.
  As above, we deduce that $\tau'=\tau$.
  If the last factor of the factorization is $a_nu_n$ then $\rho'$ ends in some 
  state $(q'_{n-1},s,p)\in F$ and $q'_{n-1}$ is the last state of $\tau'$. By 
  definition of $F$, we deduce that $\tau'$ is empty (if $n=1$) or accepting 
  (if $n>1$). Again, we obtain $\tau'=\tau$.
  
  It remains to show
  that, for each $i$, the subrun $\rho'_i$ of $\rho'$ reading $a_iu_i$
  equals $\rho_i$.  Assuming that $u_i=d_1d_2\cdots$, by definition of $\A$
  we deduce that
  $
  \rho'_i=q_{i-1}\xrightarrow{a_i}(q_{i-1},\phi(a_i),\iota_1)
  \xrightarrow{d_1}(q_{i-1},\phi(a_id_1),p'_1)
  \xrightarrow{d_2}(q_{i-1},\phi(a_id_1d_2),p'_2) \cdots
  $ 
  and $\sigma'_i=\iota_1\xrightarrow{d_1}p'_1\xrightarrow{d_2}p'_2\cdots$ is a
  run of $\A_1$.  If $u_i$ is infinite, since $\rho'$ is accepting in $\A$ we
  deduce that $\sigma'_i$ is accepting in $\A_1$.  Since $\A_1$ satisfies
  \eqref{item:G1}, we deduce that $\sigma'_i=\sigma_i$, hence also
  $\rho'_i=\rho_i$.  Assume now that $u_i$ is finite with length
  $m\geq0$.  Clearly, if $m=0$ then we have $\rho'_i=\rho_i$.  We assume
  now that $m>0$ and we show that the last state $p'_m$ of $\sigma'_i$ is final.
  If $c_i$ exists in the factorization then
  $(q_{i-1},\phi(a_iu_i),p'_m)\xrightarrow{c_i}q_i$ is a transition in $\A$ which
  implies $p'_m\in F_1\cup\{\iota_1\}$.  If the last factor is $a_iu_i$ then,
  since $\rho'$ is accepting, we deduce that $p'_m\in F_1\cup\{\iota_1\}$.  Now,
  $m>0$ and \eqref{item:G3} implies that $p_m\neq\iota_1$.  Therefore,
  $\sigma'_i$ is accepting in $\A_1$ and we deduce as above that
  $\sigma'_i=\sigma_i$, hence also $\rho'_i=\rho_i$. Since this holds 
  for all $i$'s, we have shown that $\rho'=\rho$.

\noindent{\bf {$\Aa$ satisfies \eqref{item:G2}}} 
Let $r\in Q$ be a state of $\A$ and $w\in L_r(\A)=L_{r,\da r,r}(\A)$.
  So we have in $\A$ a run $\rho=r\xrightarrow{w}r$ using intermediary states 
  strictly less than $r$. 
  
  Assume first that $r=q\in Q_2$.  Then, the run $\rho$ of $\A$ induces the
  following factorization $w=(a_1 u_1 c_1)(a_2 u_2 c_2)\cdots(a_n u_n c_n)$ with
  $n>0$.  We have
  $\rho=q\xrightarrow{a_1u_1c_1}q_1\xrightarrow{a_2u_2c_2}q_2\cdots
  q_{n-1}\xrightarrow{a_nu_nc_n}q$ and the states $q_1,\ldots,q_{n-1}$ are all
  less than $q$.  Therefore, with $b_i=\phi(a_iu_ic_i)$, we deduce that
  $\tau=q\xrightarrow{b_1}q_1\xrightarrow{b_2}q_2\cdots\xrightarrow{b_n}q$
  is a run of $\A_2$.  Since the order $<$ restricted to $Q_2$ equals $<_2$ we
  deduce that $b_1b_2\cdots b_n\in L_q(\A_2)=L_{q,\da_2 q,q}(\A_2)$.  Since
  $\A_2$ satisfies \eqref{item:G2}, we obtain $\psi(b_1\cdots b_n)=e_q$ where
  $e_q$ is the idempotent associated with state $q$ for $\A_2$.  Now,
  $\phi(w)=b_1\cdots b_n=\psi(b_1\cdots b_n)=e_q$ and we get 
  $L_q(\A)\subseteq\phi^{-1}(e_q)$.
  
  The second case is when $r=(q,s,p)\in Q_2\times S\times Q_1$.
  Since $Q_2\times S\times Q_1 < Q_2$ in $\A$, we deduce that
  $
  \rho=(q,s,p)\xrightarrow{d_1}(q,s_1,p_1)\cdots
  (q,s_{m-1},p_{m-1})\xrightarrow{d_m}(q,s,p)
  $ 
  for some $m>0$ and the intermediary states $(q,s_i,p_i)$ are all less than
  $(q,s,p)$ in $\A$.  By definition of the order $<$ in $\A$ we deduce that
  $p_i\leq_1 p$ in $\A_1$ for all $1\leq i<m$.  Therefore,
  $p\xrightarrow{d_1}p_1\cdots p_{m-1}\xrightarrow{d_m}p$ is a run of $\A_1$ and
  $w\in(L_p(\A_1))^{+}$.  Let $e_p$ be the idempotent associated with state $p$
  of $\A_1$ by \eqref{item:G2}.  We have shown that $L_r(\A)\subseteq
  (L_p(\A_1))^{+}\subseteq\phi^{-1}(e_p)$ since $e_p$ is an idempotent.  
\end{proof}

The second inductive case is when there is some semigroup element
$c\in\phi(\Sigma)$ such that $cS \subsetneq S$.  The proof is along the same
lines as the previous one but the construction turns out to be more complicated.
Again $(cS,\cdot)$ is a strict subsemigroup of $(S,\cdot)$, i.e., $|cS|<|S|$.
Let $\Sigma_2=\Sigma\cap\phi^{-1}(c)$ be the set of all letters mapped to $c$
and $\Sigma_1=\Sigma\setminus\Sigma_2$.  If $\Sigma_1=\emptyset$ then we are in
the second basic case above.  Hence we assume $\Sigma_1\neq\emptyset$ and since
$c\in\phi(\Sigma)\setminus\phi(\Sigma_1)$ we have
$|\phi(\Sigma_1)|<|\phi(\Sigma)|$ so by induction hypothesis we can construct a
\emph{good} automaton
$\A_1=(Q_1,\Sigma_1,\Delta_1,\iota_1,F_1,R_1,<_1)$ for the morphism $\phi$
restricted to $\Sigma_1$.

Each nonempty word $w$  
has a unique factorization 
$
w=a_0 u_0 (c_1 a_1 u_1)(c_2 a_2 u_2)(c_3 a_3 u_3)\cdots
$
with $a_i\in\Sigma$, $u_i\in\Sigma_1^{*}\cup\Sigma_1^{\omega}$ and
$c_i\in\Sigma_2$.  If $w\in\Sigma^{\omega}$ is infinite, the factorization has
infinitely many blocks when $w$ has infinitely many letters from $\Sigma_2$,
otherwise the factorization ends with some $a_nu_n\in\Sigma\Sigma_1^{\omega}$
with $n\geq 0$.  If the word $w\in\Sigma^{+}$ is finite then the factorization
ends with $a_nu_n\in\Sigma\Sigma_1^{*}$ with $n\geq 0$ or it ends with $c_n$
with $n\geq 1$.

We view $B=\phi(\Sigma_2\Sigma\Sigma_1^{*})\subseteq cS$ as an alphabet and we
consider the evaluation semigroup morphism $\psi\colon B^{+}\to cS$.  Let
$b_i=\phi(c_i a_i u_i)\in B$.  The factorization of $w$ yields the word
$b_1b_2b_3\cdots$ over $B$.  Moreover, for $i\leq j$ we have $\psi(b_i\cdots
b_j)=\phi(c_i a_i u_i\cdots c_j a_j u_j)$.

Since $|cS|<|S|$, we can construct a \emph{good}
automaton $\A_2=(Q_2,B,\Delta_2,\iota_2,F_2,R_2,<_2)$ for the morphism
$\psi\colon B^{+}\to cS$ by induction hypothesis.  We now show how to construct
a weakly-good automaton $\A$ for $\phi\colon\Sigma^{+}\to S$.  Intuitively, we
use $\A_1$ to scan the words $u_i$ over $\Sigma_1$ and we use $\A_2$ to scan the
sequence of blocks $c_i a_i u_i$ represented by the letters $b_i$ in $B$ (see
Figure~\ref{fig:RunInductive2}).

\begin{figure}[t]
  \centering
  \includegraphics[scale=0.23]{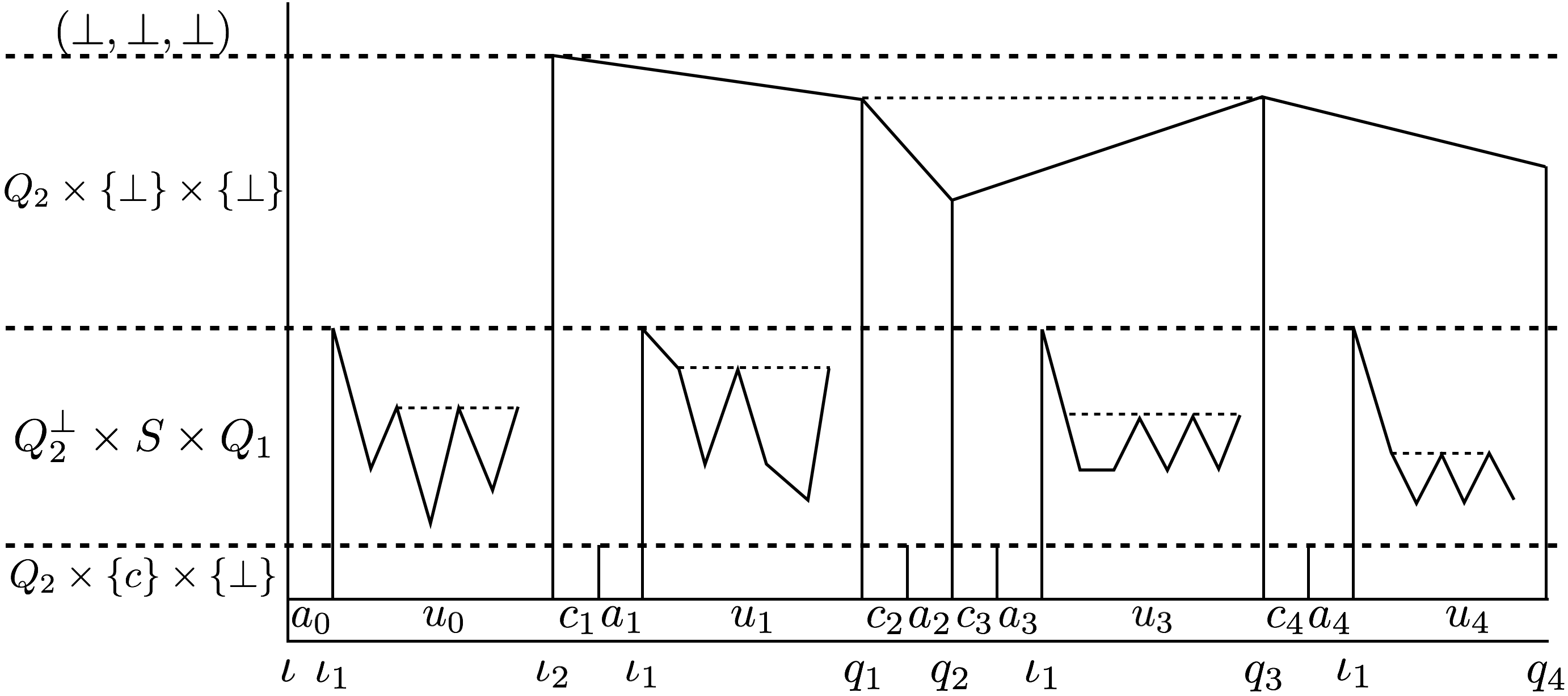}
  \caption{Run for the second inductive case: $cS\subsetneq S$.}
  \label{fig:RunInductive2}
\end{figure}

For a set $E$ and a new symbol $\bot\notin E$ we let $E^{\bot}=E\cup\{\bot\}$.
The set of states of $\A$ is $Q=Q^{\bot}_2\times S^{\bot}\times Q_1^{\bot}$. The 
initial state is $\iota=(\bot,\bot,\bot)$. The transitions are defined below 
so that:
\begin{enumerate}[nosep]
  \item  If 
  $\iota_2\xrightarrow{b_1}q_1\xrightarrow{b_2}q_2\xrightarrow{b_3}q_3\cdots$ 
  is a run of $\A_2$ then we will have in $\A$ the run
  $$
  \iota\xrightarrow{a_0u_0}(\iota_2,\bot,\bot)
  \xrightarrow{c_1a_1u_1}(q_1,\bot,\bot)
  \xrightarrow{c_2a_2u_2}(q_2,\bot,\bot)
  \xrightarrow{c_3a_3u_3}(q_3,\bot,\bot) \cdots
  $$

  \item Now, zooming in the initial factor $a_0u_0$ with $u_0=d_1d_2d_3\cdots$, if
  $\iota_1\xrightarrow{d_1}p_1\xrightarrow{d_2}p_2\xrightarrow{d_3}p_3\cdots$ 
  is a run of $\A_1$ then, we will have in $\A$ the run
  $$
  \iota
  \xrightarrow{a_0}(\bot,\phi(a_0),\iota_1)
  \xrightarrow{d_1}(\bot,\phi(a_0d_1),p_1)
  \xrightarrow{d_2}(\bot,\phi(a_0d_1d_2),p_2) 
  \xrightarrow{d_3} (\bot,\phi(a_0d_1d_2d_3),p_3) 
  \cdots
  $$

  \item Finally, zooming in some factor $c_ia_iu_i$ with $u_i=d_1d_2d_3\cdots$, if
  $\iota_1\xrightarrow{d_1}p_1\xrightarrow{d_2}p_2\xrightarrow{d_3}p_3\cdots$ 
  is a run of $\A_1$ then, with $q=q_{i-1}$, we will have in $\A$ the run
  $$
  (q,\bot,\bot)\xrightarrow{c_i}(q,c,\bot)
  \xrightarrow{a_i}(q,c\phi(a_i),\iota_1)
  \xrightarrow{d_1}(q,c\phi(a_id_1),p_1)
  \xrightarrow{d_2}(q,c\phi(a_id_1d_2),p_2) 
  \xrightarrow{d_3} %(q,c\phi(a_id_1d_2d_3),p_3) 
  \cdots
  $$
\end{enumerate}
Formally, the transitions of $\A$ are defined as follows:
\begin{enumerate}[nosep]
  \item (a) $\iota \xrightarrow{a\in\Sigma} (\bot,\phi(a),\iota_1)$,
  (b) $\iota \xrightarrow{a\in\Sigma} (\iota_2,\bot,\bot)$,

  \item $(q,\bot,\bot) \xrightarrow{a\in\Sigma_2} (q,c,\bot)$ for $q\in Q_2$, 

  \item (a) $(q,c,\bot) \xrightarrow{a\in\Sigma} (q,c\phi(a),\iota_1)$ for $q\in 
  Q_2$, (b) $(q,c,\bot) \xrightarrow{a\in\Sigma} (q',\bot,\bot)$ if
  $q\xrightarrow{c\phi(a)}q'$ in $\A_2$,

  \item  $(q,s,p) \xrightarrow{a\in\Sigma_1} (q,s\phi(a),p')$ if $q\in 
  Q_2^{\bot}$ and $p\xrightarrow{a}p'\notin F_1$ in $\A_1$,

  \item  $(q,s,p) \xrightarrow{a\in\Sigma_1} (q',\bot,\bot)$ if 
  $p\xrightarrow{a}p'\in F_1$ in $\A_1$ and 
  $q\xrightarrow{s\phi(a)}q'$ in $\A_2$ or $(q=\bot \wedge q'=\iota_2)$.
\end{enumerate}
Notice that there are non-deterministic choices between transitions of type 1(a)/1(b), 
or 3(a)/3(b) or 4/5.  Hence, even if the automata $\A_1$ and $\A_2$ are deterministic,
the automaton $\A$ constructed in this second inductive case is
non-deterministic.  We will see below that it is unambiguous.  Intuitively, the
first choice (1(a),3(a),4) has to be taken when the next letter is in $\Sigma_1$ while
the second choice (1(b),3(b),5) has to be taken when the next letter is in $\Sigma_2$.

The total order $<$ on $Q$ is defined so that 
$
Q_2\times\{c\}\times\{\bot\} <
Q^{\bot}_2\times S\times Q_1 < 
Q_2\times\{\bot\}\times\{\bot\} < \iota
$ 
and $(q,\bot,\bot)<(q',\bot,\bot)$ iff $q<_2 q'$, and $p<_1 p'$ implies
$(q,s,p)<(q,s',p')$ for all $s,s'\in S$ and $q\in Q^{\bot}_2$.  Notice that the
initial state $\iota=(\bot,\bot,\bot)$ is the maximal state in $Q$ and has no
incoming transitions, so \eqref{item:G3} holds.

The final and repeated states of $\A$ are given by 
$F =  (F_2\cup\{\iota_2\})\times\{\bot,c\}\times\{\bot\}$, and
$R = \big( R_2\times\{\bot\}\times\{\bot\} \big)
\cup \big( (F^{\bot}_2\cup\{\iota_2\})\times S\times R_1 \big)$. 

\begin{figure}[tbh]
\begin{center} 
  \scalebox{0.8}{                                                                                                                                                       
    \begin{tikzpicture}[->,thick]     
    	
  \node[initial, cir1, initial text ={}] at (-7.5,-2) (A1) {$\iota_1$} ;
  \node[cir1] at (-6.5,-2) (B1) {$p$};
  \node[cir1,accepting] at (-5.5,-2) (Z1) {$f_1$};

  \path (A1) edge node [above]{$b$}node {}(B1);
  \path (A1) edge[bend right=30] node [below]{$b$}node {}(Z1);
  \path (B1) edge[loop above] node {$b$}node {}(B1);
  \path (B1) edge node [above]{$b$}node {}(Z1);
                                                                                                                                     
  \node[initial, cir, initial text ={}] at (-7.5,-4) (A) {$\iota_2$} ;
  \node[cir] at (-6.5,-4) (B) {$q$};
  \node[cir,accepting] at (-5.5,-4) (Z) {$f_2$};
       
  \path (A) edge node [above]{$\alpha$}node {}(B);
  \path (A) edge[bend right=30] node [below]{$\alpha$}node {}(Z);
  \path (B) edge[loop above] node {$\alpha$}node {}(B);
  \path (B) edge node [above]{$\alpha$}node {}(Z);
       
  \node[initial, cir22, initial text ={}] at (-20,-4) (A2) {$(\bot, \bot, \bot)$} ;
  \node[cir22] at (-20,-2) (C) {$(\bot,\alpha,\iota_1)$};
  \node[cir22,accepting] at (-18,-4) (D) {$(\iota_2,\bot,\bot)$};
  \node[cir22] at (-20,-6) (E) {$(\bot,\beta,\iota_1)$};
      
  \node[cir22] at (-18,-2) (C1) {$(\bot,\alpha,p)$};
  \node[cir22] at (-18,-6) (E1) {$(\bot,\beta,p)$};

  \path(A2)edge node [above]{$a,b$}node {}(D);
  \path(A2)edge node [left]{$a$}node {}(C);
  \path(A2)edge node [left]{$b$}node {}(E);
       
  \node[cir22,accepting] at (-16,-4) (F) {$(\iota_2,\alpha,\bot)$};
  \node[cir22] at (-16,-6) (F1) {$(\iota_2,\alpha,\iota_1)$};
  \node[cir22] at (-15,-7.5) (F2) {$(\iota_2,\alpha,p)$};
     
  \node[cir22] at (-14,-4) (G) {$(q,\bot,\bot)$};
  \node[cir22,accepting] at (-13,-6) (G1) {$(f_2,\bot,\bot)$};
  \node[cir22,accepting] at (-11.5,-7.5) (G2) {$(f_2,\alpha,\bot)$};
  \node[cir22] at (-9,-7.5) (G3) {$(f_2,\alpha,\iota_1)$};
  \node[cir22] at (-6.5,-7.5) (G4) {$(f_2,\alpha,p)$};
      
  \node[cir22] at (-12,-4) (H) {$(q,\alpha,\bot)$};
  \node[cir22] at (-10,-4) (I) {$(q,\alpha,\iota_1)$};
  \node[cir22] at (-10,-6) (I1) {$(q,\alpha,p)$};
      
  \path(C)edge node [right]{$b$}node {}(D);
  \path(E)edge node [right]{$b$}node {}(D);
  \path(C)edge node [above]{$b$}node {}(C1);
  \path(E)edge node [below]{$b$}node {}(E1);
      
  \path(C1)edge[loop above] node [right]{$b$}node {}(C1);
  \path(E1)edge[loop below] node [right]{$b$}node {}(E1);
  \path(F2)edge[loop left] node [left]{$b$}node {}(F2);
      
  \path(C1)edge node [right]{$b$}node {}(D);
  \path(E1)edge node [right]{$b$}node {}(D);
  \path(F1)edge node [below]{$b$}node {}(G);
  \path(F1)edge node [left]{$b$}node {}(F2);
  \path(F2)edge[bend left=10] node [left]{$b$}node {}(G);
  \path(F2)edge[bend right=0] node [below]{$b$}node {}(G1);
  \path(F1)edge node [below]{$b$}node {}(G1);
  
  \path(D)edge node [above]{$a$}node {}(F);
  \path(F)edge node [above]{$a,b$}node {}(G);
  \path(F)edge node [right]{$a,b$}node {}(F1);
  \path(F)edge node [right]{$a,b$}node {}(G1);
  
  \path(G)edge node [below]{$a$}node {}(H);
  
  \path(H)edge node [left]{$a,b$}node {}(G1);
  \path(H)edge[bend right=10] node [above]{$a,b$}node {}(G);
  
  \path(H)edge node [above]{$a,b$}node {}(I);
  \path(I)edge[bend left=0] node [above]{$b$}node {}(G1);
  \path(I)edge[bend right=30] node [above]{$b$}node {}(F);
  \path(I)edge node [left]{$b$}node {}(I1);
  \path(I1)edge[loop right] node [right]{$b$}node {}(I1);
  \path(I1)edge[bend right=7] node [above]{$b$}node {}(F);
  \path(I1)edge node [below]{$b$}node {}(G1);
  \path(G1)edge[bend right=0] node [below]{$a$}node {}(G2);
  \path(G2)edge node [above]{$a,b$}node {}(G3);
  \path(G3)edge node [above]{$b$}node {}(G4);
  \path(G4)edge[loop above] node [right]{$b$}node {}(G4);
\end{tikzpicture}
}

\end{center}
 \caption{The good automata $\Aa_1$, $\Aa_2$ on the right. The weakly-good automaton $\Aa$ 
 on the left.}
 \label{fig:induct-two}	
\end{figure}
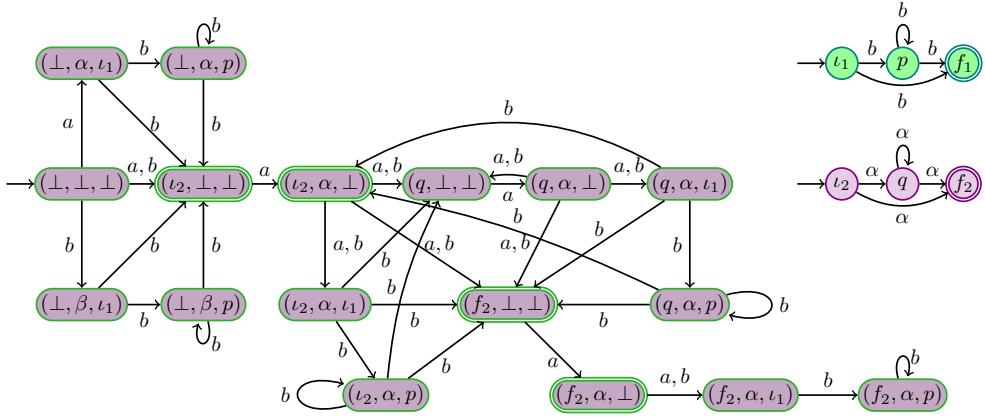

\begin{lemma}\label{lem:case2-good}
  The automaton $\A$ defined above is \emph{weakly-good} for $\phi\colon\Sigma^{+}\to S$.
\end{lemma}
\begin{proof}
  We have already seen that $\A$ satisfies \eqref{item:G3}.  
  \subsection*{$\Aa$ satisfies \eqref{item:G1}}
    Consider a word $w\in\Sigma^{+}\cup\Sigma^{\omega}$ and its
  unique factorization $w=a_0 u_0 (c_1 a_1 u_1)(c_2 a_2 u_2)(c_3 a_3 u_3)\cdots$ with
  $a_i\in\Sigma$, $u_i\in\Sigma_1^{*}\cup\Sigma_1^{\omega}$ and
  $c_i\in\Sigma_2$.  Let $b_i=\phi(c_ia_iu_i)\in B$.  There is a unique empty or
  accepting run $\tau=\iota_2\xrightarrow{b_1}q_1\xrightarrow{b_2}q_2
  \xrightarrow{b_3}q_3\cdots$ of $\A_2$. For each $i\geq 0$, assuming that 
  $u_i=d_1d_2\cdots$, there is a unique empty or accepting run 
  $\sigma_i=\iota_1\xrightarrow{d_1}p_1\xrightarrow{d_2}p_2\cdots$ of $\A_1$. We 
  construct the corresponding subruns of $\A$. When $i=0$ we define
  \begin{align*}
    \rho_0 &= \iota
    \xrightarrow{a_0}(\iota_2,\bot,\bot)
    \tag{if $u_0=\varepsilon$}
    \\
    \rho_0 &= \iota
    \xrightarrow{a_0}(\bot,\phi(a_0),\iota_1)
    \xrightarrow{d_1}(\bot,\phi(a_0d_1),p_1)
    \xrightarrow{d_2}(\bot,\phi(a_0d_1d_2),p_2) \cdots
    \\
    & \hspace{20mm} (\bot,\phi(a_0d_1\cdots d_{m-1}),p_{m-1})
    \xrightarrow{d_m}(\iota_2,\bot,\bot) \cdots
    \tag{if $|u_0|=m>0$}
    \\
    \rho_0 &= \iota
    \xrightarrow{a_0}(\bot,\phi(a_0),\iota_1)
    \xrightarrow{d_1}(\bot,\phi(a_0d_1),p_1)
    \xrightarrow{d_2}(\bot,\phi(a_0d_1d_2),p_2) \cdots
    \tag{if $u_0\in\Sigma_1^{\omega}$}
  \end{align*}
  Notice that when $|u_0|=m>0$ then the last state of $\sigma_0$ is accepting, 
  hence the last transition of $\rho_0$ in this case is well-defined.
  When $i>0$, we define (with $q_0=\iota_2$):
  \begin{align*}
    \rho_i &= (q_{i-1},\bot,\bot)
    \xrightarrow{c_i}(q_{i-1},c,\bot)
    \xrightarrow{a_i}(q_i,\bot,\bot)
    \tag{if $u_i=\varepsilon$}
    \\
    \rho_i &= (q_{i-1},\bot,\bot)
    \xrightarrow{c_i}(q_{i-1},c,\bot)
    \xrightarrow{a_i}(q_{i-1},c\phi(a_i),\iota_1)
    \xrightarrow{d_1}(q_{i-1},c\phi(a_id_1),p_1)
    \\
    &\hspace{21mm} \cdots (q_{i-1},c\phi(a_id_1\cdots d_{m-1}),p_{m-1})
    \xrightarrow{d_m} (q_i,\bot,\bot)
    \tag{if $|u_i|=m>0$}
    \\
    \rho_i &= (q_{i-1},\bot,\bot)
    \xrightarrow{c_i}(q_{i-1},c,\bot)
    \xrightarrow{a_i}(q_{i-1},c\phi(a_i),\iota_1)
    \xrightarrow{d_1}(q_{i-1},c\phi(a_id_1),p_1)
    \\
    &\hspace{21mm}\xrightarrow{d_2}(q_{i-1},c\phi(a_id_1d_2),p_2) \cdots 
    \tag{if $u_i\in\Sigma_1^{\omega}$}
  \end{align*}
  Notice that when $|u_i|=m>0$ then the last state of $\sigma_i$ is accepting
  and we have $b_i=c\phi(a_iu_i)$, hence the last transition of $\rho_i$ in
  this case is well-defined.
  
  When $w$ contains infinitely many letters from $\Sigma_2$, the factorization
  is infinite and each $u_i$ is finite.  We obtain a run
  $\rho=\rho_0\rho_1\rho_2\cdots$ of $\A$ for $w$.  Since $\tau$ is
  accepting in $\A_2$, we deduce that $\rho$ uses infinitely many states from
  $R_2\times\{\bot\}\times\{\bot\}$.  Therefore, $\rho$ is accepting in $\A$.
    
  Assume now that $w$ contains finitely many letters from $\Sigma_2$.  Then the
  factorization is finite.  If the last factor is $c_na_nu_n$ with $n>0$ then
  $\rho=\rho_0\rho_1\rho_2\cdots\rho_n$ is a run of $\A$ for $w$. 
  \begin{itemize}[nosep]
    \item If $u_n$ is finite then the last state of $\rho$ is
    $r=(q_n,\bot,\bot)$.  Since $\tau$ is accepting, we deduce that $q_n\in F_2$
    and therefore $r\in F$ and $\rho$ is accepting.

    \item If $u_n\in\Sigma_1^{\omega}$ is infinite then the run $\tau$ of 
    $\A_2$ ends in state $q_{n-1}$ (recall that $q_0=\iota_2$). Since $\tau$ is 
    empty or accepting, we have $q_{n-1}\in F_2\cup\{\iota_2\}$. Now, the run 
    $\sigma_n$ of $\A_1$ reading $u_n$ is accepting, hence it uses infinitely 
    many states from $R_1$. We deduce that $\rho_n$ uses infinitely many states 
    in $\{q_{n-1}\}\times S\times R_1$ and $\rho$ is accepting. 
  \end{itemize}
  If the last factor is $c_n$ with $n>0$ then the run $\tau$ of $\A_2$ ends in
  state $q_{n-1}\in F_2\cup\{\iota_2\}$.  Therefore,
  $\rho=\rho_0\rho_1\rho_2\cdots\rho_{n-1}\xrightarrow{c_n}(q_{n-1},c,\bot)$ is
  a run of $\A$ for $w$ which is accepting by definition of $F$.
  
  The last case is when $w=a_0u_0$ with
  $u_0\in\Sigma_1^{*}\cup\Sigma_1^{\omega}$.  Then $\rho=\rho_0$ is a run of
  $\A$ for $w$.  If $u_0$ is finite, then $\rho=\rho_0$ ends in state
  $(i_2,\bot,\bot)\in F$ and $\rho$ is accepting.  If $u_0$ is infinite, then
  $\rho=\rho_0$ uses infinitely many states in $\{\bot\}\times S\times R_1$
  since $\sigma_0$ is accepting in $\A_1$.  Again, $\rho$ is accepting.

  We have proved that the automaton $\A$ accepts all words in
  $\Sigma^{+}\cup\Sigma^{\omega}$.
  
  \medskip
  We show now that $\A$ is unambiguous.  Let $\rho'$ be an accepting run of $\A$
  on $w$.  We have to show that $\rho'=\rho$ where $\rho$ is the accepting run
  for $w$ defined above.  By definition of $\A$, the run $\rho'$ induces the
  very same factorization of $w=a_0 u_0 (c_1 a_1 u_1)(c_2 a_2 u_2) \cdots$ with
  $a_i\in\Sigma$, $u_i\in\Sigma_1^{*}\cup\Sigma_1^{\omega}$ and
  $c_i\in\Sigma_2$.  Moreover, we can write
  $$
  \rho'=\iota\xrightarrow{a_0u_0}(\iota_2,\bot,\bot)
  \xrightarrow{c_1a_1u_1}(q'_1,\bot,\bot)
  \xrightarrow{c_2a_2u_2}(q'_2,\bot,\bot)\cdots
  $$
  We denote by $\rho'_0$ the subrun of $\rho'$ reading $a_0u_0$ and by $\rho'_i$
  the subrun of $\rho'$ reading $c_ia_iu_i$ for $i>0$.  From the definition of
  the transistions in $\A$, it is easy to check that
  $\tau'=\iota_2\xrightarrow{b_1}q'_1\xrightarrow{b_2}q'_2\cdots$ is a run of
  $\A_2$.  We first show that $\tau'=\tau$.
  
  If $w$ has infinitely many letters from $\Sigma_2$ then the run $\tau'$ is
  infinite and none of the states $q'_i$ belongs to $F_2\cup\{\iota_2\}$ since
  $\A_2$ is good.  Now $\rho'$ is accepting in $\A$ and by definition of $R$ we
  deduce that $q'_i\in R_2$ for infinitely many $i$'s.  Therefore $\tau'$ is
  accepting in $\A_2$.  Since $\A_2$ satisfies \eqref{item:G1}, we deduce that
  $\tau'=\tau$, i.e., $q'_i=q_i$ for all $i$.
  
  If $w$ has finitely many letters from $\Sigma_2$ and the last factor is $c_n$
  with $n>0$ then $\rho'$ ends in state $(q'_{n-1},c,\bot)\in F$.  We deduce
  that $q'_{n-1}\in F_2\cup\{\iota_2\}$ and $\tau'$ is empty (if $n=1$) or
  accepting (if $n>1$).  As above, we deduce that $\tau'=\tau$.  
  
  Assume now that $w$ has finitely many letters from $\Sigma_2$ and the 
  factorization ends with $a_nu_n$ ($n\geq0$). If $n=0$ then $\tau'$ is empty 
  and we get $\tau'=\tau$. So we assume $n>0$. If $u_n$ is infinite, then by 
  definition of $R$ we have $q'_{n-1}\in F_2^{\bot}\cup\{\iota_2\}$ (with 
  $q'_0=\iota_2$). We deduce that $\tau'$ is empty when 
  $n=1$ or accepting ending with $q'_{n-1}\in F_2$ when $n>1$. Again we deduce 
  that $\tau'=\tau$.
  The last case is when $u_n$ is finite. Since $\rho'$ is accepting, it ends in 
  some state $r\in F$. Due to the letter $a_n$, $r=(q,c,\bot)$ is not possible.
  Therefore, $r=(q'_n,\bot,\bot)\in F$ and $q'_n\in F_2\cup\{\iota_2\}$. We 
  deduce that $\tau'$ is accepting and again $\tau'=\tau$.
  
  It remains to show that $\rho'_i=\rho_i$ for all $i$.  Assume that
  $u_i=d_1d_2\cdots$.  We start with the case $i=0$.  There are three cases
  depending on whether $u_0$ is empty, finite of length $m>0$, or infinite.
  
  If $u_0=\varepsilon$, by definition of $\A$ we deduce that
  $\rho'_0=\iota\xrightarrow{a_0}(\iota_2,\bot,\bot)$. Indeed, either $w=a_0$ 
  and $\rho'_0=\rho'$ which is accepting, which implies that the last state of 
  $\rho'$ is $(\iota_2,\bot,\bot)$ by definition of $F$. Or the letter $c_1$ 
  exists and the second transition of $\rho'$ must be of type 3, which implies 
  again that the first transition of $\rho'$ is of type 2. In both cases,
  $\rho'_0=\iota\xrightarrow{a_0}(\iota_2,\bot,\bot)=\rho_0$.
  
  If $u_0$ is of length $m>0$ then by definition of $\A$ we deduce that
  $$
  \rho'_0=\iota\xrightarrow{a_0}(\bot,\phi(a_0),\iota_1)
  \xrightarrow{d_1}(\bot,\phi(a_0d_1),p'_1) \cdots
  (\bot,\phi(a_0d_1\cdots d_{m-1}),p'_{m-1}) \xrightarrow{d_m} (\iota_2,\bot,\bot)
  $$
  As above, we can check that the last state of $\rho'_0$ must be
  $(\iota_2,\bot,\bot)$ either because $w=a_0u_0$ and $\rho'$ is accepting, or
  because $c_1$ exists and the transition reading $c_1$ must start from
  $(\iota_2,\bot,\bot)$.  Therefore, the last transition of $\rho'_0$ is of type
  7 and we deduce that $p'_{m-1}\xrightarrow{d_m}p'_m\in F_1$ in $\A_1$ (with 
  $p'_0=\iota_1$).
  Therefore, $\sigma'_0=\iota_1\xrightarrow{d_1}p'_1\cdots
  p'_{m-1}\xrightarrow{d_m}p'_m$ is an accepting run of $\A_1$ for $u_0$. Since 
  $\A_1$ is unambiguous, we obtain $\sigma'_0=\sigma_0$ and it follows 
  $\rho'_0=\rho_0$.

  If $u_0$ is infinite, by definition of $\A$ we deduce that
  $$
  \rho'_0=\iota\xrightarrow{a_0}(\bot,\phi(a_0),\iota_1)
  \xrightarrow{d_1}(\bot,\phi(a_0d_1),p'_1)
  \xrightarrow{d_2}(\bot,\phi(a_0d_1d_2),p'_2) \cdots
  $$
  and $\sigma'_0=\iota_1\xrightarrow{d_1}p'_1\xrightarrow{d_2}p'_2\cdots$ is a
  run of $\A_1$ for $u_0$.  Since $\rho'$ is accepting in $\A$ we deduce that
  $\sigma'_0$ is accepting in $\A_1$.  Since $\A_1$ is unambiguous,
  we deduce that $\sigma'_0=\sigma_0$, hence also $\rho'_0=\rho_0$.

  The case $i>0$ is handled similarly. Again, there are three cases
  depending on whether $u_i$ is empty, finite of length $m>0$, or infinite.
  
  If $u_i=\varepsilon$, by definition of $\A$ we deduce that
  $\rho'_i=(q_{i-1},\bot,\bot)\xrightarrow{c_i}(q_{i-1},c,\bot)
  \xrightarrow{a_i}(q_i,\bot,\bot)$.  Indeed, either $c_ia_i$ is the last factor
  of $w$ and since $\rho'$ is accepting, the last state of $\rho'_i$ which is
  also the last state of $\rho'$ must be $(q_i,\bot,\bot)$ by definition of $F$.
  Or the letter $c_{i+1}$ exists and the transition of $\rho'$ reading $c_{i+1}$
  must be of type 3, which implies again that the last transition of $\rho'_i$
  is of type 5.  In both cases,
  $\rho'_i=\rho_i$.
  
  If $u_i$ is of length $m>0$ then by definition of $\A$ we deduce that
  \begin{align*}
    \rho'_i=(q_{i-1},\bot,\bot)\xrightarrow{c_i}(q_{i-1},c,\bot)
    \xrightarrow{a_i}(q_{i-1},c\phi(a_i),\iota_1)
    \xrightarrow{d_1}(q_{i-1},c\phi(a_id_1),p'_1) &\cdots
    \\
    (q_{i-1},\phi(a_id_1\cdots d_{m-1}),p'_{m-1}) &\xrightarrow{d_m} (q_i,\bot,\bot)
  \end{align*}
  As above, we can check that the last state of $\rho'_i$ must be
  $(q_i,\bot,\bot)$ either because $c_ia_iu_i$ is the last factor of $w$ and
  $\rho'$ is accepting, or because $c_{i+1}$ exists and the transition reading
  $c_{i+1}$ must start from $(q_i,\bot,\bot)$.  Therefore, the last transition
  of $\rho'_i$ is of type 7 and we deduce that $p'_{m-1}\xrightarrow{d_m}p'_m\in
  F_1$ in $\A_1$ (with $p'_0=\iota_1$).  Therefore,
  $\sigma'_i=\iota_1\xrightarrow{d_1}p'_1\cdots p'_{m-1}\xrightarrow{d_m}p'_m$
  is an accepting run of $\A_1$ for $u_i$.  Since $\A_1$ is unambiguous, we
  obtain $\sigma'_i=\sigma_i$ and it follows $\rho'_i=\rho_i$.

  If $u_i$ is infinite, by definition of $\A$ we deduce that
  $$
    \rho'_i=(q_{i-1},\bot,\bot)\xrightarrow{c_i}(q_{i-1},c,\bot)
    \xrightarrow{a_i}(q_{i-1},c\phi(a_i),\iota_1)
    \xrightarrow{d_1}(q_{i-1},c\phi(a_id_1),p'_1) 
    \xrightarrow{d_2}
    \cdots
  $$
  and $\sigma'_i=\iota_1\xrightarrow{d_1}p'_1\xrightarrow{d_2}p'_2\cdots$ is a
  run of $\A_1$ for $u_i$.  Since $\rho'$ is accepting in $\A$ we deduce that
  $\sigma'_i$ is accepting in $\A_1$.  Since $\A_1$ is unambiguous,
  we deduce that $\sigma'_i=\sigma_i$, hence also $\rho'_i=\rho_i$.
  
  We have proved that $\rho'_i=\rho_i$ for all $i$. Therefore,  
  $\rho'=\rho'_0\rho'_1\rho'_2\cdots =\rho_0\rho_1\rho_2\cdots =\rho$.
  
  \subsection*{$\Aa$ satisfies \eqref{item:G2}}
  Let $r\in Q$ be a state
  of $\A$ and $w\in L_r(\A)=L_{r,\da r,r}(\A)$.  So we have in $\A$ a run
  $\rho=r\xrightarrow{w}r$ using intermediary states strictly less than $r$.
  Note that $r=\iota=(\bot,\bot,\bot)$ is not possible.  Also, $r\in
  Q_2\times\{c\}\times\{\bot\}$ is not possible since between two occurrences of
  such states, we must use a transition of type 3, hence we must have a state
  from $Q_2\times\{\bot\}\times\{\bot\}$ which is strictly above for $<$ in
  $\A$.
  
  Assume first that $r=(q,\bot,\bot)$ with $q\in Q_2$.  Then, the run $\rho$ of
  $\A$ induces the following factorization $w=(c_1 a_1 u_1)(c_2 a_2 u_2)\cdots
  (c_n a_n u_n)$ with $n>0$.  We have
  $$
  \rho=(q,\bot,\bot)\xrightarrow{c_1a_1u_1}(q_1,\bot,\bot)
  \xrightarrow{c_2a_2u_2}(q_2,\bot,\bot)\cdots
  (q_{n-1},\bot,\bot)\xrightarrow{c_na_nu_n}(q,\bot,\bot)
  $$
  and the intermediary states $(q_i,\bot,\bot)$ are all less than
  $(q,\bot,\bot)$ in $\A$.  Therefore, with $b_i=\phi(c_ia_iu_i)$, we deduce
  that $\tau=q\xrightarrow{b_1}q_1\xrightarrow{b_2}q_2\cdots
  q_{n-1}\xrightarrow{b_n}q$ is a run of $\A_2$.  By definition of the order
  $<$, we deduce that $q_i<_2 q$ for $1\leq i<n$. Therefore, $b_1b_2\cdots b_n\in
  L_q(\A_2)=L_{q,\da_2 q,q}(\A_2)$.  Since $\A_2$ satisfies \eqref{item:G2}, we
  obtain $\psi(b_1\cdots b_n)=e_q$ where $e_q$ is the idempotent associated with
  state $q$ for $\A_2$.  Now, $\phi(w)=b_1\cdots b_n=\psi(b_1\cdots b_n)=e_q$
  and we get $L_r(\A)\subseteq\phi^{-1}(e_q)$.
  
  The second case is when $r=(q,s,p)\in Q^{\bot}_2\times S\times Q_1$.  Since
  $Q^{\bot}_2\times S\times Q_1 < Q_2\times\{\bot\}\times\{\bot\}$ in $\A$, the 
  run $\rho$ may only use transitions of type 6. We
  deduce that
  $$
  \rho=(q,s,p)\xrightarrow{d_1}(q,s_1,p_1)\cdots
  (q,s_{m-1},p_{m-1})\xrightarrow{d_m}(q,s,p)
  $$
  for some $m>0$ and the intermediary states $(q,s_i,p_i)$ are all less than
  $(q,s,p)$ in $\A$.  By definition of the order $<$ in $\A$ we deduce that
  $p_i\leq_1 p$ in $\A_1$ for all $1\leq i<m$.  Therefore,
  $p\xrightarrow{d_1}p_1\cdots p_{m-1}\xrightarrow{d_m}p$ is a run of $\A_1$ and
  $w\in(L_p(\A_1))^{+}$.  Let $e_p$ be the idempotent associated with state $p$
  of $\A_1$ by \eqref{item:G2}.  We have shown that $L_r(\A)\subseteq
  (L_p(\A_1))^{+}\subseteq\phi^{-1}(e_p)$ since $e_p$ is an idempotent.  This
  concludes the proof.
\end{proof}

\begin{example}\label{eg-case2}
  To illustrate the second inductive case, consider the morphism $\varphi\colon
  \Sigma^+ \to S=\{\alpha,\beta\}$ defined in Example~\ref{ex:eg-good}. 
  The first inductive case does not apply to $\varphi$ since $S\alpha=S\beta=S$.  The
  second inductive case can be used here: $\alpha S=\{\alpha\} \subsetneq S$.  Then,
  $\Sigma_1=\{b\}$ and $\Sigma_2=\{a\}$.  We have the morphism $\varphi_1\colon
  \Sigma_1^{+} \to \{\beta\}$ to which, applying the first or second basic case
  gives us the automaton $\Aa_1$.  Also, $B=\varphi(\Sigma_2\Sigma\Sigma_1^*)=\{\alpha\}$ and we
  have the morphism $\psi \colon B^+ \to \alpha S=\{\alpha\}$.  Both $\Aa_1$, $\Aa_2$ are
  automata with three states (see Figure~\ref{fig:induct-two}).  We can apply the
  construction explained above on $\Aa_1$ and $\Aa_2$ to obtain $\Aa$ as in
  Figure~\ref{fig:induct-two}.  Notice that the automaton $\Aa$ in
  Figure~\ref{fig:eg-good} (manually crafted) is also $\varphi$-good.

  \noindent{\bf{Why the weakly-good $\Aa$ cannot be deterministic.}} On this example, we now
  explain why a weakly-good automaton $\Aa$ for $\phi$ must be
  non-deterministic.  Towards a contradiction, assume that there exists a
  $\varphi$-weakly-good deterministic automaton.  Let $\Aa$ be such an automaton
  with a minimal number of states.  Let $q$ be the highest ranked state of $\Aa$
  reachable from $\iota$ such that $L_{q} \neq \emptyset$.  Consider $v \in
  L_q$.  Without loss of generality, assume that $v \in a \Sigma^*$.  Then, we
  claim that $L_q \cap b\Sigma^*=\emptyset$.  If not, we will have words $v \in
  a \Sigma^*$ and $v' \in b \Sigma^*$ both in $L_q$.  Since $\Aa$ satisfies
  \eqref{item:G2}, we know that $L_q \subseteq \varphi^{-1}(e)$ for some
  idempotent $e$.  The claim follows since $\varphi(a \Sigma^*) = \{\alpha\}$,
  $\varphi(b \Sigma^*) = \{\beta\}$ and $\alpha \neq \beta$.  Thus $L_q \subseteq a
  \Sigma^*$.  Since $\Aa$ is deterministic and universal \eqref{item:G1}, we have
  an outgoing transition on $b$ from $q$.  Let $q \stackrel{b}{\rightarrow} q'$.
  Then $q$ cannot be reached from $q'$.  Indeed, assume there is a run
  $\rho=q'\xrightarrow{w}q$ from $q'$ to $q$ and let $q''$ be the highest state
  in $\rho$.  From the run $q'\xrightarrow{w}q\xrightarrow{b}q'\xrightarrow{w}q$ we deduce that
  $L_{q''}\neq\emptyset$, which implies $q''\leq q$ by choice of $q$.  From the
  run $q\xrightarrow{b}q'\xrightarrow{w}q$ we deduce that $L_q \cap b \Sigma^* \neq
  \emptyset$, a contradiction.  This means that in any run of $\Aa$, $\iota
  \stackrel{u}{\rightarrow} q \stackrel{v}{\rightarrow} q
  \stackrel{b}{\rightarrow} q' \dots$, there is no occurrence of $q$ after $q'$.
  This allows us to construct an automaton $\Aa'$ from $\Aa$ with a new initial
  state $\iota'$ having transitions $\delta_{\Aa'}(\iota',x)=\delta_{\Aa}(q',x)$
  for $x \in \{a,b\}$, and $\delta_{\Aa'}(r,x)=\delta_{\Aa}(r,x)$ for all
  $r\notin\{\iota,q\}$ and $x \in \{a,b\}$.  We can check that $\Aa'$ is
  deterministic and $\varphi$-weakly-good, and has at least one less state than
  $\Aa$.  This contradicts the minimality of $\Aa$.
 
  Notice that if we allow a look-ahead of size 1, we can obtain a
  $\phi$-weakly-good deterministic automaton; the automaton $\Aa$ of
  Figure~\ref{fig:eg-good} (middle) is in fact one such.
  The good automaton $\Bb$ on the right Figure~\ref{fig:eg-good} can also be
  converted into a weakly-good, deterministic automaton with look-ahead two.
  Notice that we can generalize this example to show that in general, to obtain
  weakly-good and deterministic automata, a bounded look-ahead will not suffice.
Below, we generalize this argument.
\end{example}

\section*{The inherent non-determinism of weakly-good automata}
\label{app:non-det}
In this section, we show that the $\varphi$-weakly-good automata cannot be made
deterministic even with bounded look-ahead.  The second inductive case is the
reason why this cannot be.  If we look at the construction in the second
inductive case, we give a split of $w \in \Sigma^+$ into chunks of $\Sigma_2
\Sigma \Sigma_1^*$.  Even if $\Aa_1, \Aa_2$ are both deterministic and complete
weakly-good, $\Aa$ introduces non-determinism since we have to guess whether the
next symbol is in $\Sigma_2$ or in $\Sigma_1$, each time we process $u_i \in
\Sigma_1^*$.  This may give an impression that we can get rid of the
non-determinism by using a look-ahead of size 1, which simply checks if the next
symbol is in $\Sigma_1$ or $\Sigma_2$.  While this is true for
Example~\ref{eg-case2} (see the left and middle automata in
Figure~\ref{fig:eg-good}), in general it is not possible to construct a
$\varphi$-weakly-good automaton which is deterministic, and has a bounded
look-ahead.

Consider the morphism $\varphi: \Sigma^+ \rightarrow S$ where $\Sigma=\{a,b\}$,
$S=\Sigma^{\leq k}=\{u\in\Sigma^{+}\mid |u|\leq k\}$, $\varphi(x)=x$ for all
$x\in\Sigma$ and the product in $S$ is so that the elements in $\Sigma^{k}$ are
right-absorbant: $\alpha\cdot\beta=\alpha$ for all $\alpha\in\Sigma^{k}$ and
$\beta\in S$.  The morphism $\varphi$ is a generalization of the morphism in
Example~\ref{eg-case2}.  Notice that the idempotents of $S$ are all elements of
$\Sigma^k$.  It is easy to see that one can construct a $\varphi$-weakly-good
automaton which is deterministic with a $k$-look-ahead (generalizing
Figure~\ref{fig:eg-good}).  We show that it is not possible to construct a
$\varphi$-weakly-good automaton $\Aa$ which is deterministic with a
$(k-1)$-look-ahead.  Let us assume that we can indeed do this, and let $\Aa$ be
such an automaton with a minimal number of states.
   
Since $\Aa$ satisfies \eqref{item:G2}, we know that for each state $q$, there is
an idempotent $v_q\in\Sigma^{k}$ such that
$L_q\subseteq\varphi^{-1}(e_q)=v_q\Sigma^{*}$.
Let $q$ be the highest ranked state of $\Aa$ which occurs at least twice on some
infinite accepting run $\rho$ of $\Aa$.  We may write
$\rho=\iota\xrightarrow{u_1}q\xrightarrow{v_1au_2}q\xrightarrow{v_2w}$ where
$u_1,u_2\in\Sigma^{*}$, $v_1,v_2\in\Sigma^{k-1}$, $a\in\Sigma$ and
$w\in\Sigma^{\omega}$.  The unique accepting run on $u_1(v_1au_2)(v_1au_2)v_2w$
must start with
$\iota\xrightarrow{u_1}q\xrightarrow{v_1au_2}q\xrightarrow{v_1au_2}q\cdots$
since $\A$ is deterministic with $(k-1)$-look-ahead.  By choice of $q$ we deduce
that all states $q'$ occurring in the subrun $q\xrightarrow{v_1au_2}q$ satisfy
$q'\leq q$.  We deduce that $v_1au_2\in L_q^{+}\subseteq v_q\Sigma^{*}$ and
therefore $v_1a=v_q$.

Since $\Aa$ is universal \eqref{item:G1} and deterministic with
$(k-1)$-look-ahead, there are accepting runs for all words in
$u_1v_1au_2v_2b\Sigma^{\omega}$ and all these runs start with
$\iota\xrightarrow{u_1}q\xrightarrow{v_1au_2}q$.  Along these runs, the state
$q$ cannot be reached again.  Otherwise, we would have a run
$\iota\xrightarrow{u_1}q\xrightarrow{v_1au_2}q\xrightarrow{v_2bu_3}q\xrightarrow{v_3w}$.
As above, we would get $v_2bu_3\in L_q^{+}\subseteq v_q\Sigma^{*}$.  This is a
contradiction since $v_q=v_1a\neq v_2b$.

This allows us to construct from $\Aa$, an automaton $\Aa'$ as
follows.  Let $\iota'$ be the initial state of $\Aa'$, and define
$\delta_{\Aa'}(\iota', x?v)=\delta_{\Aa}(q, v_2bx?v)$, for $x \in \Sigma, v \in
\Sigma^{k-1}$.  Note that $\delta_{\Aa'}(\iota', x?v)$ is a state of $\Aa$ other
than $q$.  Further, $\delta_{\Aa'}(p,x?v)=\delta_{\Aa}(p,x?v)$ for all $p \neq
q, x \in \Sigma, v \in \Sigma^{k-1}$.  This makes $\Aa'$ a strictly smaller
deterministic, $\varphi$-weakly-good automaton with $(k-1)$-look-ahead whenever
$\Aa$ is $\varphi$-weakly-good, contradicting the minimality of $\Aa$.

\begin{remark}
  Notice that if we are dealing with commutative semigroups, then $Sc=cS$ for
  any $c \in S$.  In this case, the second inductive case $cS \subsetneq S$
  coincides with the first one.  The difficulty occurs when dealing with
  non-commutative semigroups, and in this case, the proof is much more
  challenging for the case $cS \subsetneq S$ as seen above.
\end{remark}

\subsection*{Wrapping Up}
Now we show that we have covered all cases.  Let $\phi\colon\Sigma^{+}\to S$ be
a semigroup morphism such that for all $c\in\phi(\Sigma)$ we have $cS=S=Sc$, 
i.e., neither of the two inductive cases may be applied. 
Wlog, we assume that $\phi(\Sigma^{+})=S$, otherwise we restrict $S$ to its 
sub-semigroup $\phi(\Sigma^{+})$. Hence, each element $s\in S$ can be written 
as a product $s=c_1\cdots c_k$ where $c_1,\ldots,c_k\in\phi(\Sigma)$. From the 
hypothesis it follows that $sS=S=Ss$ for all $s\in S$. Using 
Lemma~\ref{lem:semi-group} below we deduce that $S$ is a group so that we are 
in the first basic case (Lemma~\ref{lem:group}).
Lemma \ref{lem:semi-group} is a folklore result, and also works 
for infinite semi-groups. 
\begin{lemma}\label{lem:semi-group}
  If $S$ is a finite semigroup such that $sS=S=Ss$ for all $s$ in $S$ then $S$
  is a group.
\end{lemma}
\begin{proof}
  We show first that $S$ contains a unit element and next that all elements 
  have an inverse. 
  
  Since $S$ is a finite semigroup, it contains some idempotent $e$.  Now, $Se=S$
  implies that the right multiplication by $e$ defines a permutation $\sigma_e$
  of $S$.  We obtain $\sigma_e=\sigma_e\circ\sigma_e$ since $e$ is an
  idempotent.  We deduce that $\sigma_e=\mathsf{Id}$ is the identitiy since
  permutations of $S$ with composition form a group.  Therefore,
  $s=\sigma_e(s)=se$ for all $s\in S$ and $e$ is a right unit.  Using $eS=S$ we
  deduce similarly that $e$ is a left unit and therefore a unit of $S$.
  
  Finally, let $s\in S$. From $Ss=S=sS$ we deduce that $rs=e=st$ for some 
  $r,t\in S$. %, i.e., $s$ admits a left-inverse $r$ and a right inverse $t$.
  It follows that $r=re=r(st)=(rs)t=et=t$ which is the inverse of $s$.
\end{proof}

\section{Applications}

We now focus on two applications obtained from synthesizing good automata.
Given a morphism $\varphi\colon \Sigma^+ \to S$ for a semi-group $S$, we first
derive the forest factorization theorem from the $\varphi$-good automaton $\Aa$
constructed above.

\begin{theorem}[Forest Factorization Derived]\label{thm:split}	
  Let $\varphi\colon\Sigma^+\to S$ be a morphism.
  For each finite or infinite word $w\in\Sigma^{\infty}$, we can construct a 
  Ramsey split $\sigma$ whose height is bounded by the number of states of a 
  weakly-good automaton for $\phi$.
\end{theorem}

\begin{proof}
  Let $\Aa$ be a weakly-good automaton for the morphism $\varphi$.  In
  particular $\Aa$ satisfies \ref{item:G1} and \ref{item:G2}.  Let $h\colon(Q,<)
  \to (\{1, \dots, |Q|\},<)$ be a monotone bijection.  To define the split
  $\sigma$ of $w=a_1a_2a_3\cdots\in\Sigma^{\infty}$, consider the unique
  accepting run
  $\rho=q_0\xrightarrow{a_1}q_1\xrightarrow{a_2}q_2\xrightarrow{a_3}q_3\cdots$
  of $w$ in $\Aa$ and define $\sigma(i)=h(q_i)$ for all positions $i\geq0$ of
  $w$.  Notice that two positions $i<j$ are $\sigma$-equivalent ($i\sim j$) iff
  $q_i=q_j$ and $q_k\leq q_i$ for all $i\leq k\leq j$.  We deduce that
  $w(i,j]=a_{i+1}\cdots a_j\in L_{q_i}^{+}$.  Hence, $\phi(w(i,j])=e_{q_i}$ is
  the idempotent associated with state $q_i$.  Therefore, the split $\sigma$ for
  $w$ is Ramsey. 
\end{proof}
\subsection*{Some remarks on the height of the factorisation tree}
\label{app:height}
Theorem \ref{thm:split} gives an easy proof for the existence of a Ramsey split
where the height is bounded by the number of states of the weakly-good automaton
for $\varphi$.  Notice that this bound on the height is rather loose and can be
optimized.  To get an idea of this height $H$, we look at the basic and
inductive cases.  In the first base case when $S$ is a group, we know by
construction that $H=|S|+1$.  In the second basic case where
$|\varphi(\Sigma)|=1$, we know that $|S|=k+\ell-1$, and the automaton had $k+n$
states where $k \leq n \leq k+\ell$, obtaining $H \leq 2|S|$.

Now let us turn to the inductive cases.  Let $H_1$, $H_2$ respectively be the
number of states of $\Aa_1$ and $\Aa_2$.  The monotone bijection $h$ for $\Aa$
is defined using the monotone bijections $h_1\colon (Q_1, <_1) \to (\{1,\dots,
|Q_1|\}, <_1)$ and $h_2\colon (Q_2, <_2) \to (\{1,\dots, |Q_2|\}, <_2)$ obtained
from $\Aa_1, \Aa_2$.

Assuming we are in the first inductive case, the number of states $|Q|$ of the
constructed $\Aa$ is $|Q|=|Q_2|+|Q_2|\times |S| \times |Q_1|$.  Actually, one
can check that we can save on the height of the split by defining
$h(q)=H_1+h_2(q)$ for $q\in Q_2$, and $h((q,s,p))=h_1(p)$ for $(q,s,p)\in Q_2
\times S \times Q_1$. The map $h$ is not a bijection anymore, but a careful 
analysis shows that the split as defined in the proof of 
Theorem~\ref{thm:split} is Ramsey.

Now assume we are in the second inductive case.  The number of states $|Q|$ of
the constructed $\Aa$ is $|Q|=(|Q_2|+1)\times (|S|+1) \times (|Q_1|+1)$.  Since
the states of $Q_2 \times \{c\} \times \bot$ are the lowest in the ordering, and
since we know that we cannot revisit any $(q,c,\bot)$ without seeing an higher
state, we can safely assign the same height to all of them: $h(Q_2 \times \{c\} 
\times \bot)=1$. As in the first inductive case, we can also define 
$h(q,s,p)=1+h_1(p)$ for $(q,s,p)\in Q_2^{\bot}\times S\times Q_1$,
$h(q,\bot,\bot)=1+H_1+h_2(q)$ for $q\in Q_2$ and $h(\bot, \bot, \bot)=H_1+H_2+2$.
The split obtained in this way is Ramsey.
  
Note that in both cases, $h$ is indeed monotone, respecting the ordering of
states in $\Aa$ (see Figures \ref{fig:RunInductive1}, \ref{fig:RunInductive2}).
Moreover, the bound on the height $H$ that we require is $\leq H_1+H_2+2$.

\subsection{Good Automata to Good Expressions}

In this section, we show how we can use the $\varphi$-good automata to obtain
good expressions.  We start from a \emph{good} automaton
$\A=(Q,\Sigma,\Delta,\iota,f,R,<)$ for a semigroup morphism
$\phi\colon\Sigma^{+}\to S$.  Wlog, we assume that $\Aa$ is reduced,
i.e., all states in $\Aa$ belong to some accepting run.  We construct the good
expressions by state elimination.  For all $p,q\in Q$ and $X\subseteq Q$ such
that $X<\{p,q\}$, and for all $s\in S$, we construct a $\phi$-good expression
$F_{p,X,q}^{s}$ such that $\Lang{F_{p,X,q}^{s}}=L_{p,X,q}\cap\phi^{-1}(s)$
(recall that $\emptyset$ is a good expression).  The construction is by
induction.

The base case is when $X=\emptyset$.  Then,
$L_{p,\emptyset,q}\subseteq\Sigma$ so $F^{s}_{p,\emptyset,q}$ is
either empty or a finite union of letters from $\Sigma$, which is indeed 
$\phi$-good.

Let $r\in Q$, $X=\da r=\{r'\in Q\mid r'<r\}$ and $Y=X\cup\{r\}$.  Assume by
induction that for all $p,q$ such that $X<\{p,q\}$ and all $s\in S$ we have
already constructed good expressions $F_{p,X,q}^{s}$. In particular, we have 
already computed the good expressions $F_{r,X,r}^{s}$. Since the automaton $\A$ 
is \emph{good}, we have $L_{r,X,r}=L_r\subseteq\phi^{-1}(e_r)$ where $e_r\in S$ 
is the idempotent associated with $r$. In particular, $L_{r,X,r}^{s}=\emptyset$ if 
$s\neq e_r$.

Let $p,q\in Q$ be such that $Y<\{p,q\}$ and let $s\in S$. We define
\begin{align*}
  F^{s}_{p,Y,q} & = F^{s}_{p,X,q} \cup
  \bigcup_{s_1s_2=s} F^{s_1}_{p,X,r}\cdot F^{s_2}_{r,X,q}
  \cup
  \bigcup_{s_1e_rs_2=s} (F^{s_1}_{p,X,r}\cdot (F_{r,X,r}^{e_r})^{+}) \cdot F^{s_2}_{r,X,q}
  \,.
\end{align*}

\begin{lemma}
  The regular expression $F_{p,Y,q}^{s}$ is \emph{good} and we have
  $\Lang{F_{p,Y,q}^{s}}=L_{p,Y,q}\cap\phi^{-1}(s)$.
\end{lemma}

\begin{proof}
  First, notice that the expression $F^{s}_{p,Y,q}$ is unambiguous since  the 
  automaton $\Aa$ is reduced and unambiguous.
  By inductive hypothesis, $F^{s}_{p,X,q}$, $F^{s_1}_{p,X,r}$, $F^{s_2}_{r,X,q}$
  and $F_{r,X,r}^{e_r}$ are good expressions.  Since $e_r$ is an idempotent,
  $(F_{r,X,r}^{e_r})^+$ is also a good expression.  We can also check that each 
  sub-expression maps to the same semigroup element, which could be $s$ or some 
  $s_1e_r$ in the last union. In particular, we have $F_{p,Y,q}^{s}\in\varphi^{-1}(s)$.
  We deduce that $F^{s}_{p,Y,q}$ is good.
  The fact that $\Lang{F_{p,Y,q}^{s}} \subseteq L_{p,Y,q}$ follows easily.
\end{proof}

\begin{proof}[Proof of Theorem~\ref{thm:U-forest}]
  \eqref{item:T1} Recall that since the automaton $\A$ is \emph{good}, it satisfies 
  \eqref{item:G4}. Let $X=Q\setminus\{\iota,f\}$. We have $X<\{\iota,f\}$.
  We let $F_s=F_{\iota,X,f}^{s}$. We have 
  $\Lang{F_s}=L_{\iota,X,f}\cap\phi^{-1}(s)=\Sigma^{+}\cap\phi^{-1}(s)$.

  \eqref{item:T2}
  Each $\omega$-word $w\in\Sigma^{\omega}$ has a unique accepting run in $\A$.
  We partition $\Sigma^{\omega}$ according to the largest reapeated state along
  an accepting run.  Let $r\in R$ be a repeated (B\"uchi) state and recall that
  $\da r=\{r'\in Q\mid r'<r\}$.  The $\omega$-words accepted by $\A$ using $r$
  as the largest accepting state can be described with the $\omega$-regular
  expression $F_{\iota,\da r,r}\cdot(F_{r,\da r,r}^{e_r})^{\omega}$.  Therefore the
  unambiguous expression is
  $$
  G=\bigcup_{r\in R} F_{\iota,\da r,r}\cdot(F_{r,\da r,r}^{e_r})^{\omega} \,.
  $$
  This conclude the proof.
\end{proof}

\bibliography{UFF}

\end{document}